\documentclass{article}
\usepackage{graphicx}
\usepackage{subfigure}
\usepackage{dsfont}
\usepackage{stmaryrd}
\usepackage{ifthen}
\usepackage{color}
\usepackage{algorithm2e}
\usepackage{epstopdf}
\usepackage[utf8]{inputenc}
\usepackage[T1]{fontenc} 
\usepackage[english]{babel}
\usepackage{amsmath,amssymb,latexsym,euscript}
\usepackage[standard]{ntheorem}
\usepackage{pstricks,pst-node,pst-coil}
\usepackage{subfigure}

\usepackage{url,tikz}
\usepackage{pgflibrarysnakes}

\usepackage{multicol}

\usetikzlibrary{arrows}
\usetikzlibrary{automata}
\usetikzlibrary{snakes}
\usetikzlibrary{shapes}

\definecolor{bleuclair}{rgb}{0.75,0.75,1.0}

\newcommand {\tv}[1] {\lVert #1 \rVert_{\rm TV}}

\def \P {\mathbb{P}}
\def \ov {\overline}
\def \ts {\tau_{\rm stop}}

\newenvironment{xpl}[1][Running Example]{\textbf{#1.} }{\ \rule{0.5em}{0.5em}}

\newtheorem{prpstn}{Proposition}
\newtheorem{thrm}{Theorem}
\newtheorem{crllr}{Corollary}
\newtheorem{lmm}{Lemma}

\newcommand{\Nats}[0]{\ensuremath{\mathbb{N}}}
\newcommand{\R}[0]{\ensuremath{\mathbb{R}}}

\newcommand{\Bool}[0]{\ensuremath{\mathds{B}}}

\newcommand{\CG}[1]{\begin{color}{blue}\textit{}\end{color}}

\newcommand{\fg}[1]{''}

\begin{document}
%
\title{Random Walk in a N-cube Without Hamiltonian Cycle 
  to Chaotic Pseudorandom Number Generation: Theoretical and Practical 
  Considerations}

%
%
%
\author{Sylvain Contassot-Vivier, Jean-François Couchot, Christophe Guyeux, and Pierre-Cyrille Heam}

\maketitle

\begin{abstract}
  Designing a pseudorandom number generator (PRNG) is a
difficult  and complex  task.  Many  recent works  have considered  chaotic
functions as the basis of built PRNGs: the quality of the output would
indeed
be an obvious consequence of some chaos properties.  However, there is
no  direct  reasoning that  goes  from  chaotic functions  to  uniform
distribution of the output.  
Moreover,  
embedding such kind of functions into a PRNG does not necessarily
allow to get a  chaotic output,
which could be required for simulating some chaotic behaviors.

In a  previous work,  some of  the authors have  proposed the  idea of
walking into  a $\mathsf{N}$-cube  where a balanced  Hamiltonian cycle
has been removed as the basis of  a chaotic PRNG. In this article, all
the  difficult  issues  observed  in   the  previous  work  have  been
tackled.  The  chaotic behavior  of  the  whole  PRNG is  proven.  The
construction of  the balanced  Hamiltonian cycle is  theoretically and
practically solved. An upper bound of  the expected length of the walk
to  obtain a  uniform distribution  is calculated.   Finally practical
experiments show  that the generators successfully  pass the classical
statistical tests.
\end{abstract}

\section{Introduction}
The exploitation of chaotic systems to generate pseudorandom sequences
is  a very topical issue~\cite{915396,915385,5376454}.  Such   systems  are
fundamentally chosen  because of  their unpredictable character  and their
sensitiveness to initial conditions.   In most cases, these generators
simply  consist in  iterating  a chaotic  function  like the  logistic
map~\cite{915396,915385} or  the Arnold's  one~\cite{5376454}\ldots 
Optimal  parameters of  such functions remain to be found so that
attractors are avoided,\textit{e.g.}.
By following this procedure, generated numbers will hopefully 
follow a uniform  distribution.  In order to check the  quality of the
produced outputs, PRNGs (Pseudo-Random Number
Generators)    are usually  tested with   statistical    batteries   like   the   so-called
DieHARD~\cite{Marsaglia1996},          NIST~\cite{Nist10},          or
TestU01~\cite{LEcuyerS07} ones.

In its  general understanding,  the notion of chaos  is  often reduced  to the
strong  sensitiveness  to  the  initial  conditions  (the  well  known
``butterfly effect''): a continuous function $k$ defined on a metrical
space is said to be \emph{strongly sensitive to the  initial conditions} if
for each point $x$ and each  positive value $\epsilon$, it is possible
to find another point $y$ as close  as possible to $x$, and an integer
$t$ such that the distance between the $t$-th iterates of $x$ and $y$,
denoted by $k^t(x)$ and $k^t(y)$, is larger than $\epsilon$. However,
in  his definition  of  chaos, Devaney~\cite{Devaney}  imposes to  the
chaotic function  two other properties called  \emph{transitivity} and
\emph{regularity}. The functions mentioned above  have been studied according
to these  properties, and they  have been  proven as chaotic  on $\R$.
But  nothing  guarantees  that  such  properties  are  preserved  when
iterating the functions on floating point numbers, which is the domain
of interpretation of real numbers $\R$ on machines.

To avoid this  lack of chaos, we have previously  presented some PRNGs
that iterate  continuous functions $G_f$  on a discrete domain  $\{ 1,
\ldots,  n  \}^{\Nats}  \times  \{0,1\}^n$, where  $f$  is  a  Boolean
function  (\textit{i.e.},  $f  :  \{0,1\}^{\mathsf{N}}
\rightarrow  \{0,1\}^{\mathsf{N}}$).
These         generators          are         $\textit{CIPRNG}_f^1(u)$
\cite{guyeuxTaiwan10,bcgr11:ip},            $\textit{CIPRNG}_f^2(u,v)$
\cite{wbg10ip},             and             $\chi_{\textit{14Secrypt}}$
\cite{DBLP:conf/secrypt/CouchotHGWB14}    where   \textit{CI} stands for
\emph{Chaotic Iterations}.  We have firstly proven in~\cite{bcgr11:ip}
that,    to    establish    the   chaotic    nature    of    
$\textit{CIPRNG}_f^1$ algorithm,  it  is  necessary   and  sufficient  that  the
asynchronous iterations  are strongly  connected. We then  have proven
that it is necessary and  sufficient that the Markov matrix associated
to  this graph  is  doubly  stochastic, in  order  to  have a  uniform
distribution of  the outputs.  We have finally  established sufficient
conditions to guarantee the first  property of connectivity. Among the
generated   functions,   we   thus   have   considered   for   further
investigations  only  the  ones   that  satisfy  the  second  property
as well.

However,      it     cannot      be     directly      deduced     that
$\chi_{\textit{14Secrypt}}$ is chaotic since we  do not output all the
successive values of iterating $G_f$.   This algorithm only displays a
subsequence $x^{b.n}$  of a whole  chaotic sequence $x^{n}$ and  it is
indeed incorrect to say that the chaos property is preserved for any
subsequence of  a chaotic sequence.  This  article presents conditions
to preserve this property.

Finding a Boolean function which provides a  strongly connected
iteration graph having a doubly stochastic Markov matrix is however
not an easy task.
We have firstly proposed in~\cite{bcgr11:ip} a  generate-and-test based
approach that solved this issue. However, this one was not efficient enough.
Thus, a second scheme has been further presented
in~\cite{DBLP:conf/secrypt/CouchotHGWB14} by remarking that
a  $\mathsf{N}$-cube where an Hamiltonian cycle (or equivalently a Gray code)
has been removed is strongly connected and has
a doubly stochastic Markov matrix.

However, the removed Hamiltonian cycle  
has a great influence in the quality of the output.
For instance, if this one is not balanced (\textit{i.e.},
the number of changes in different bits are completely different),
some bits would be hard to switch.
This article shows an effective algorithm that efficiently 
implements the previous scheme and thus provides  functions issued
from removing, in the $\mathsf{N}$-cube, a \emph{balanced} Hamiltonian cycle. 

The length $b$ of the walk to reach a
distribution close to the uniform one would be dramatically long.
This article theoretically and practically
studies the length $b$ until the corresponding Markov
chain is close to the uniform distribution.
Finally, the ability of the approach to face classical tests
suite is evaluated.

This article, which 
is an extension of~\cite{DBLP:conf/secrypt/CouchotHGWB14}, 
is organized  as  follows. The  next
section   is   devoted   to  preliminaries,   basic   notations,   and
terminologies   regarding   Boolean    map   iterations.    Then,   in
Section~\ref{sec:proofOfChaos},  Devaney's  definition   of  chaos  is
recalled  while the  proof  of chaos  of our  most  general PRNGs  is
provided.      This    is     the     first    major     contribution.
Section~\ref{sec:SCCfunc}  recalls a general scheme
to obtain functions with an expected behavior. Main theorems are recalled
to make the article self-sufficient. 
The  next  section (Sect.~\ref{sec:hamilton}) presents an algorithm that
implements this scheme and proves that it always produces a solution.  
This  is   the   second   major  contribution.
Then, Section~\ref{sec:hypercube} defines the theoretical framework to study
the  mixing-time,  \textit{i.e.},  
the sufficient amont of time until  reaching  an uniform
distribution. It proves that this one is in the worst case quadratic in the number
of elements. Experiments show that the bound is in practice 
significantly lower. This  is   the   third   major  contribution.  
Section~\ref{sec:prng} gives practical results  on evaluating the PRNG
against  the NIST  suite.  This  research  work ends  with a  conclusion
section, where the contribution is summarized and intended future work
is outlined.


\section{Preliminaries}\label{sec:preliminaries}
In what follows, we consider the Boolean algebra on the set 
$\Bool=\{0,1\}$ with the classical operators of conjunction '.', 
of disjunction '+', of negation '$\overline{~}$', and of 
disjunctive union $\oplus$. 

Let us first introduce basic notations.
Let $\mathsf{N}$ be a positive integer. The set $\{1, 2, \hdots , \mathsf{N}\}$
of integers belonging between $1$ and $\mathsf{N}$
is further denoted as $\llbracket 1, \mathsf{N} \rrbracket$.
A  {\emph{Boolean map} $f$ is 
a function from $\Bool^{\mathsf{N}}$  
to itself 
such that 
$x=(x_1,\dots,x_{\mathsf{N}})$ maps to $f(x)=(f_1(x),\dots,f_{\mathsf{N}}(x))$.
In what follows, for any finite set $X$, $|X|$ denotes its cardinality and 
$\lfloor y \rfloor$ is
the largest integer lower than $y$.

Functions are iterated as follows. 
At the $t^{th}$ iteration, only the $s_{t}-$th component is said to be
``iterated'', where $s = \left(s_t\right)_{t \in \mathds{N}}$ is a sequence of indices taken in $\llbracket 1;{\mathsf{N}} \rrbracket$ called ``strategy''. 
Formally,
let $F_f:  \Bool^{{\mathsf{N}}} \times \llbracket1;{\mathsf{N}}\rrbracket$ to $\Bool^{\mathsf{N}}$ be defined by
\[
F_f(x,i)=(x_1,\dots,x_{i-1},f_i(x),x_{i+1},\dots,x_{\mathsf{N}}).
\]
Then, let $x^0\in\Bool^{\mathsf{N}}$ be an initial configuration
and $s\in
\llbracket1;{\mathsf{N}}\rrbracket^\Nats$ be a strategy, 
the dynamics are described by the recurrence
\begin{equation}\label{eq:asyn}
x^{t+1}=F_f(x^t,s_t).
\end{equation}

Let be given a Boolean map $f$. Its associated   
{\emph{iteration graph}}  $\Gamma(f)$ is the
directed graph such that  the set of vertices is
$\Bool^{\mathsf{N}}$, and for all $x\in\Bool^{\mathsf{N}}$ and $i\in \llbracket1;{\mathsf{N}}\rrbracket$,
the graph $\Gamma(f)$ contains an arc from $x$ to $F_f(x,i)$.
Each arc $(x,F_f(x,i))$ is labelled with $i$.

\begin{xpl}
Let us consider for instance ${\mathsf{N}}=3$.
Let 
$f^*: \Bool^3 \rightarrow \Bool^3$ be defined by
$f^*(x_1,x_2,x_3) = 
(x_2 \oplus x_3, \overline{x_1}\overline{x_3} + x_1\overline{x_2},
\overline{x_1}\overline{x_3} + x_1x_2)$.
The iteration graph $\Gamma(f^*)$ of this function is given in 
Figure~\ref{fig:iteration:f*}.
\end{xpl}

\begin{figure}[ht]
\begin{center}
\includegraphics[width=0.45\textwidth]{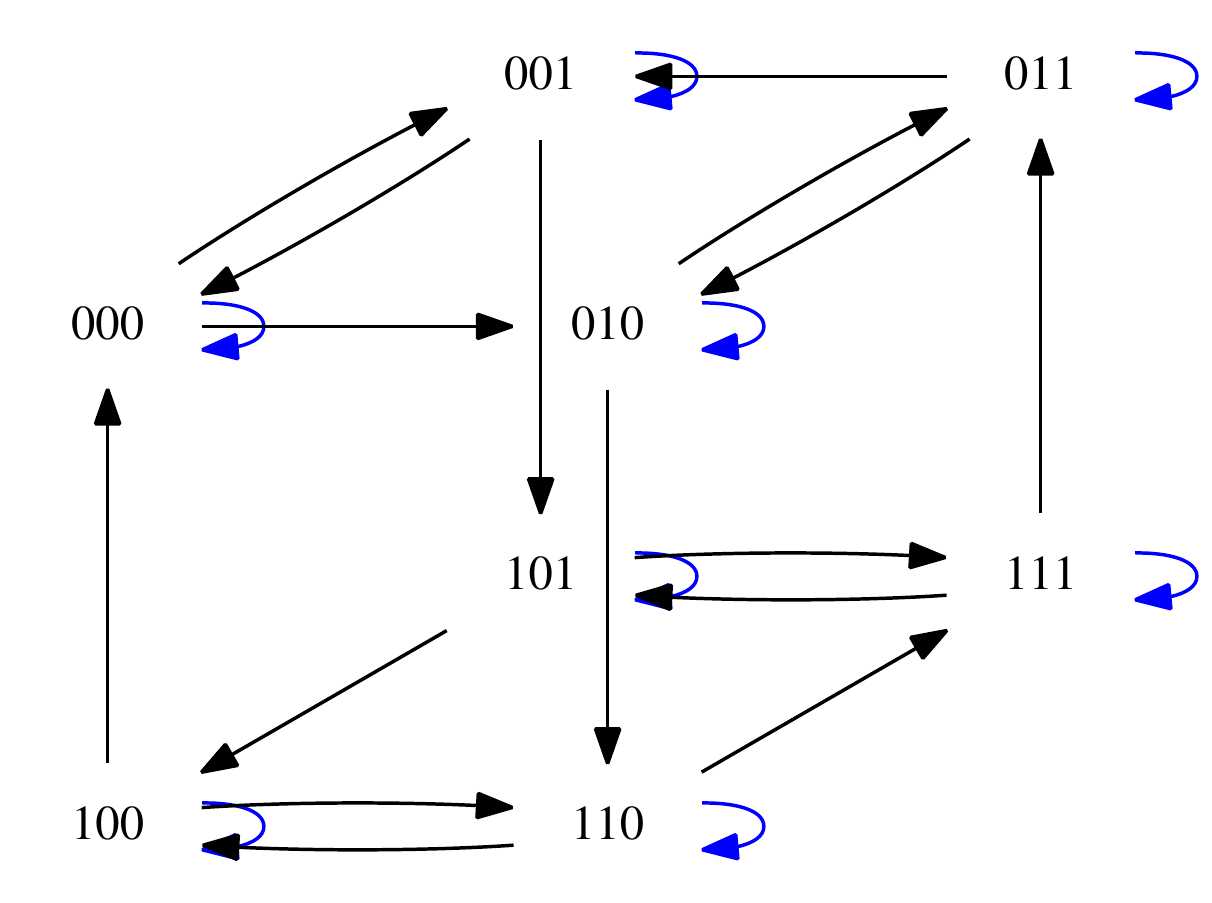}
\end{center}
\caption{Iteration Graph $\Gamma(f^*)$ of the function $f^*$}\label{fig:iteration:f*}
\end{figure}

Let us finally recall the pseudorandom number generator $\chi_{\textit{14Secrypt}}$
\cite{DBLP:conf/secrypt/CouchotHGWB14}
formalized in Algorithm~\ref{CI Algorithm}.
It is based on random walks in $\Gamma(f)$. 
More precisely, let be given a Boolean map $f:\Bool^{\mathsf{N}} \rightarrow \Bool^{\mathsf{N}}$,
an input PRNG \textit{Random},
an integer $b$ that corresponds to a number of iterations, and 
an initial configuration $x^0$. 
Starting from $x^0$, the algorithm repeats $b$ times 
a random choice of which edge to follow and traverses this edge.
The final configuration is thus outputted.

\begin{algorithm}[ht]
\begin{scriptsize}
\KwIn{a function $f$, an iteration number $b$, an initial configuration $x^0$ (${\mathsf{N}}$ bits)}
\KwOut{a configuration $x$ (${\mathsf{N}}$ bits)}
$x\leftarrow x^0$\;
\For{$i=0,\dots,b-1$}
{
$s\leftarrow{\textit{Random}({\mathsf{N}})}$\;
$x\leftarrow{F_f(x,s)}$\;
}
return $x$\;
\end{scriptsize}
\caption{Pseudo Code of the $\chi_{\textit{14Secrypt}}$ PRNG}
\label{CI Algorithm}
\end{algorithm}

Based on this setup,
we can study the chaos properties of these 
functions.
This is the aim of the next section.

\section{Proof of Chaos}\label{sec:proofOfChaos}

\subsection{Motivations}
Let us us first recall the chaos theoretical context presented 
in~\cite{bcgr11:ip}. In this article, the space of interest 
is $\Bool^{{\mathsf{N}}} \times \llbracket1;{\mathsf{N}}\rrbracket^{\Nats}$ 
and the iteration function $\mathcal{H}_f$ is  
the map from 
$\Bool^{{\mathsf{N}}} \times \llbracket1;{\mathsf{N}}\rrbracket^{\Nats}$ 
to itself defined by
\[
\mathcal{H}_f(x,s)=(F_f(x,s_0),\sigma(s)).
\] 
In this definition, 
$\sigma: \llbracket1;{\mathsf{N}}\rrbracket^{\Nats} \longrightarrow
 \llbracket1;{\mathsf{N}}\rrbracket^{\Nats} 
$
 is a shift operation on sequences (\textit{i.e.}, a function that removes the 
first element of the sequence) formally defined with
$$
\sigma((u^k)_{k \in \Nats}) =  (u^{k+1})_{k \in \Nats}. 
$$

We have proven~\cite[Theorem 1]{bcgr11:ip} that   
$\mathcal{H}_f$ is chaotic in 
$\Bool^{{\mathsf{N}}} \times \llbracket1;{\mathsf{N}}\rrbracket^{\Nats}$
if and only if $\Gamma(f)$ is strongly connected.
However, the corollary which would say that $\chi_{\textit{14Secrypt}}$ is chaotic 
cannot be directly deduced since we do not output all the successive
values of iterating $F_f$. Only a few of them are concerned and 
any subsequence of a chaotic sequence  is   not  necessarily  
a   chaotic  sequence  as well.
This necessitates a rigorous proof, which is the aim of this section.
Let us firstly recall the theoretical framework in which this research takes place.

\subsection{Devaney's Chaotic Dynamical Systems}
\label{subsec:Devaney}

Consider a topological space $(\mathcal{X},\tau)$ and a continuous function $f :
\mathcal{X} \rightarrow \mathcal{X}$~\cite{Devaney}.

\begin{definition}
The function $f$ is said to be \emph{topologically transitive} if, for any pair of open sets
$U,V \subset \mathcal{X}$, there exists $k>0$ such that $f^k(U) \cap V \neq
\varnothing$.
\end{definition}

\begin{definition}
An element $x$ is a \emph{periodic point} for $f$ of period $n\in \mathds{N}^*$
if $f^{n}(x)=x$.
\end{definition}

\begin{definition}
$f$ is said to be \emph{regular} on $(\mathcal{X}, \tau)$ if the set of periodic
points for $f$ is dense in $\mathcal{X}$: for any point $x$ in $\mathcal{X}$,
any neighborhood of $x$ contains at least one periodic point (without
necessarily the same period).
\end{definition}

\begin{definition}[Devaney's formulation of chaos~\cite{Devaney}]
The function $f$ is said to be \emph{chaotic} on $(\mathcal{X},\tau)$ if $f$ is regular and
topologically transitive.
\end{definition}

The chaos property is strongly linked to the notion of ``sensitivity'', defined
on a metric space $(\mathcal{X},d)$ by:

\begin{definition}
\label{sensitivity} The function $f$ has \emph{sensitive dependence on initial conditions}
if there exists $\delta >0$ such that, for any $x\in \mathcal{X}$ and any
neighborhood $V$ of $x$, there exist $y\in V$ and $n > 0$ such that
$d\left(f^{n}(x), f^{n}(y)\right) >\delta $.

The constant $\delta$ is called the \emph{constant of sensitivity} of $f$.
\end{definition}

Indeed, Banks \emph{et al.} have proven in~\cite{Banks92} that when $f$ is
chaotic and $(\mathcal{X}, d)$ is a metric space, then $f$ has the property of
sensitive dependence on initial conditions (this property was formerly an
element of the definition of chaos).

\subsection{A Metric Space for PRNG Iterations}

Let us first introduce $\mathcal{P} \subset \mathds{N}$ a finite nonempty
set having the cardinality $\mathsf{p} \in \mathds{N}^\ast$.
Intuitively, this  is the set of authorized numbers of iterations.
Denote by $p_1, p_2, \hdots, p_\mathsf{p}$ the ordered elements of $\mathcal{P}$: $\mathcal{P} = \{ p_1, p_2, \hdots, p_\mathsf{p}\}$
and $p_1< p_2< \hdots < p_\mathsf{p}$. 

In our Algorithm~\ref{CI Algorithm}, 
$\mathsf{p}$ is 1 and $p_1$ is $b$. 
But this algorithm can be seen as $b$ functional compositions of $F_f$.
Obviously, it can be generalized with $p_i$, $p_i \in \mathcal{P}$,
functional compositions of $F_f$.
Thus, for any $p_i \in \mathcal{P}$ we introduce the function 
$F_{f,p_i} :  \mathds{B}^\mathsf{N} \times \llbracket 1, \mathsf{N} \rrbracket^{p_i}  \rightarrow \mathds{B}^\mathsf{N}$ defined by 
\[
\begin{array}{l}
F_{f,p_i} (x,(u^0, u^1, \hdots, u^{p_i-1}))  \mapsto \\
\qquad F_f(\hdots (F_f(F_f(x,u^0), u^1), \hdots), u^{p_i-1}).
\end{array}
\]

The considered space is 
 $\mathcal{X}_{\mathsf{N},\mathcal{P}}=  \mathds{B}^\mathsf{N} \times \mathds{S}_{\mathsf{N},\mathcal{P}}$, where 
$\mathds{S}_{\mathsf{N},\mathcal{P}}=
\llbracket 1, \mathsf{N} \rrbracket^{\Nats}\times 
\mathcal{P}^{\Nats}$. 
Each element in this space is a pair where the first element is 
$\mathsf{N}$-uple in $\Bool^{\mathsf{N}}$, as in the previous space.  
The second element is a pair $((u^k)_{k \in \Nats},(v^k)_{k \in \Nats})$ of infinite sequences.
The sequence $(v^k)_{k \in \Nats}$ defines how many iterations are executed at time $k$ before the next output, 
while $(u^k)_{k \in \Nats}$ details which elements are modified. 

Let us introduce the shift function $\Sigma$ for any element of $\mathds{S}_{\mathsf{N},\mathcal{P}}$.

\[
\begin{array}{cccc}
\Sigma:&\mathds{S}_{\mathsf{N},\mathcal{P}} &\rightarrow
&\mathds{S}_{\mathsf{N},\mathcal{P}} \\
& \left((u^k)_{k \in \mathds{N}},(v^k)_{k \in \mathds{N}}\right) & \mapsto & 
\begin {array}{l}
    \left(\sigma^{v^0}\left((u^k)_{k \in \mathds{N}}\right), \right. \\
     \qquad \left. \sigma\left((v^k)_{k \in \mathds{N}}\right)\right).
 \end{array} 
\end{array}
\]

In other words, $\Sigma$ receives two sequences $u$ and $v$, and
it operates $v^0$ shifts on the first sequence and a single shift
on the second one. 
Let us consider
\begin{equation}
\begin{array}{cccc}
G_f :&  \mathcal{X}_{\mathsf{N},\mathcal{P}} & \rightarrow & \mathcal{X}_{\mathsf{N},\mathcal{P}}\\
   & (e,(u,v)) & \mapsto & \left( F_{f,v^0}\left( e, (u^0, \hdots, u^{v^0-1}\right), \Sigma (u,v) \right) .
\end{array}
\end{equation}
Then the outputs $(y^0, y^1, \hdots )$ produced by the $\textit{CIPRNG}_f^2(u,v)$ generator~\cite{wbg10:ip} 
are by definition the first components of the iterations $X^0 = (x^0, (u,v))$ and $\forall n \in \mathds{N}, 
X^{n+1} = G_f(X^n)$ on $\mathcal{X}_{\mathsf{N},\mathcal{P}}$.
The new obtained generator can be shown as either a post-treatment over generators $u$ and $v$, or a
discrete dynamical system on a set constituted by binary vectors and couple of integer sequences.

\subsection{A metric on $\mathcal{X}_{\mathsf{N},\mathcal{P}}$}

We define a distance $d$ on $\mathcal{X}_{\mathsf{N},\mathcal{P}}$ as follows. 
Consider 
$x=(e,s)$ and $\check{x}=(\check{e},\check{s})$ in 
$\mathcal{X}_{\mathsf{N},\mathcal{P}} = \mathds{B}^\mathsf{N} \times \mathds{S}_{\mathsf{N},\mathcal{P}} $,
where $s=(u,v)$ and $\check{s}=(\check{u},\check{v})$ are in $ \mathds{S}_{\mathsf{N},\mathcal{P}} = 
\mathcal{S}_{\llbracket 1, \mathsf{N} \rrbracket} \times \mathcal{S}_\mathcal{P}$. 
\begin{itemize}
\item $e$ and $\check{e}$ are integers belonging in $\llbracket 0, 2^{\mathsf{N}-1} \rrbracket$. The Hamming distance
on their binary decomposition, that is, the number of dissimilar binary digits, constitutes the integral
part of $d(X,\check{X})$.
\item The fractional part is constituted by the differences between $v^0$ and $\check{v}^0$, followed by the differences
between finite sequences $u^0, u^1, \hdots, u^{v^0-1}$ and  $\check{u}^0, \check{u}^1, \hdots, \check{u}^{\check{v}^0-1}$, followed by 
 differences between $v^1$ and $\check{v}^1$, followed by the differences
between $u^{v^0}, u^{v^0+1}, \hdots, u^{v^1-1}$ and  $\check{u}^{\check{v}^0}, \check{u}^{\check{v}^0+1}, \hdots, \check{u}^{\check{v}^1-1}$, etc.
More precisely, let $p = \lfloor \log_{10}{(\max{\mathcal{P}})}\rfloor +1$ and $n = \lfloor \log_{10}{(\mathsf{N})}\rfloor +1$.
\begin{itemize}
\item The $p$ first digits of $d(x,\check{x})$ are $|v^0-\check{v}^0|$ written in decimal numeration (and with $p$ digits: zeros are added on the left if needed).
\item The next $n\times \max{(\mathcal{P})}$ digits aim at measuring how much $u^0, u^1, \hdots, u^{v^0-1}$ differ from $\check{u}^0, \check{u}^1, \hdots, \check{u}^{\check{v}^0-1}$. The $n$ first
digits are $|u^0-\check{u}^0|$. They are followed by 
$|u^1-\check{u}^1|$ written with $n$ digits, etc.
\begin{itemize}
\item If
$v^0=\check{v}^0$, then the process is continued until $|u^{v^0-1}-\check{u}^{\check{v}^0-1}|$ and the fractional
part of $d(X,\check{X})$ is completed by 0's until reaching
$p+n\times \max{(\mathcal{P})}$ digits.
\item If $v^0<\check{v}^0$, then the $ \max{(\mathcal{P})}$  blocs of $n$
digits are $|u^0-\check{u}^0|$, ..., $|u^{v^0-1}-\check{u}^{v^0-1}|$,
$\check{u}^{v^0}$ (on $n$ digits), ..., $\check{u}^{\check{v}^0-1}$ (on $n$ digits), followed by 0's if required.
\item The case $v^0>\check{v}^0$ is dealt similarly.
\end{itemize}
\item The next $p$ digits are $|v^1-\check{v}^1|$, etc.
\end{itemize}
\end{itemize}

This distance has been defined to capture all aspects of divergences between two sequences generated 
by the $\textit{CIPRNG}_f^2$ method, when setting respectively $(u,v)$ and $(\check{u},\check{v})$ as inputted 
couples of generators. The integral part measures the bitwise Hamming distance between 
the two $\mathsf{N}$-length binary vectors chosen as seeds. The fractional part must decrease 
when the number of identical iterations applied by the $\textit{CIPRNG}_f^2$ discrete dynamical system on these seeds, in both cases (that is, when inputting either $(u,v)$ or $(\check{u},\check{v})$), increases.
More precisely, the fractional part will alternately measure the following elements:
\begin{itemize}
  \item Do we iterate the same number of times between the next two outputs, when considering either $(u,v)$ or $(\check{u},\check{v})$?
  \item Then, do we iterate the same components between the next two outputs of $\textit{CIPRNG}_f^2$ ?
  \item etc.
\end{itemize}
Finally, zeros are put to be able to recover what occurred at a given iteration. Such aims are illustrated in the two following examples.
\begin{xpl}
Consider for instance that $\mathsf{N}=13$, $\mathcal{P}=\{1,2,11\}$ (so $\mathsf{p}=3$, $p=\lfloor \log_{10}{(\max{\mathcal{P}})}\rfloor +1 = 2$, while $n=2$), and that
$s=\left\{
\begin{array}{l}
u=\underline{6,} ~ \underline{11,5}, ...\\
v=1,2,...
\end{array}
\right.$
while
$\check{s}=\left\{
\begin{array}{l}
\check{u}=\underline{6,4} ~ \underline{1}, ...\\
\check{v}=2,1,...
\end{array}
\right.$.

So 
$$d_{\mathds{S}_{\mathsf{N},\mathcal{P}}}(s,\check{s}) = 0.01~0004000000000000000000~01~1005...$$
Indeed, the $p=2$ first digits are 01, as $|v^0-\check{v}^0|=1$, 
and we use $p$ digits to code this difference ($\mathcal{P}$ being $\{1,2,11\}$, this difference can be equal to 10). We then take the $v^0=1$ first terms of $u$, each term being coded in $n=2$ digits, that is, 06. As we can iterate
at most $\max{(\mathcal{P})}$ times, we must complete this
value by some 0's in such a way that the obtained result
has $n\times \max{(\mathcal{P})}=22$ digits, that is: 
0600000000000000000000. Similarly, the first $\check{v}^0=2$ 
terms in $\check{u}$ are represented by 0604000000000000000000, and the value of their
digit per digit absolute difference is equal to 0004000000000000000000. These digits are concatenated to 01, and
we start again with the remainder of the sequences.
\end{xpl}

\begin{xpl}
Consider now that $\mathsf{N}=9$ ($n=1$), $\mathcal{P}=\{2,7\}$ ($\mathsf{p}=2, p=1$), and that

$s=\left\{
\begin{array}{l}
u=\underline{6,7,} ~ \underline{4,2,} ...\\
v=2,2,...
\end{array}
\right.$\\
while
$\check{s}=\left\{
\begin{array}{l}
\check{u}=\underline{4, 9, 6, 3, 6, 6, 7,} ~ \underline{9, 8}, ...\\
\check{v}=7,2,...
\end{array}
\right.$

So: 
$d_{\mathds{S}_{\mathsf{N},\mathcal{P}}}(s,\check{s}) = 0.5~2263667~1~5600000...$. 
\end{xpl}

$d$ can be more rigorously written as follows:
$$d(x,\check{x})=d_{\mathds{S}_{\mathsf{N},\mathcal{P}}}(s,\check{s})+d_{\mathds{B}^\mathsf{N}}(e,\check{e}),$$
where: 
\begin{itemize}
\item $d_{\mathds{B}^\mathsf{N}}$ is the Hamming distance,
\item $\forall s=(u,v), \check{s}=(\check{u},\check{v}) \in \mathcal{S}_{\mathsf{N},\mathcal{P}}$,\newline 
\[
\begin{array}{l}
 d_{\mathds{S}_{\mathsf{N},\mathcal{P}}}(s,\check{s}) = \\
  \quad  \sum_{k=0}^\infty \dfrac{1}{10^{(k+1)p+kn\max{(\mathcal{P})}}} 
   \bigg(|v^k - \check{v}^k|  \\
   \quad\quad + \left| \sum_{l=0}^{v^k-1} 
       \dfrac{u^{\sum_{m=0}^{k-1} v^m +l}}{ 10^{(l+1)n}} -
       \sum_{l=0}^{\check{v}^k-1} 
       \dfrac{\check{u}^{\sum_{m=0}^{k-1} \check{v}^m +l}}{ 10^{(l+1)n}} \right| \bigg)
\end{array}
\]
\end{itemize}

Let us show that,
\begin{prpstn}
$d$ is a distance on $\mathcal{X}_{\mathsf{N},\mathcal{P}}$.
\end{prpstn}

\begin{proof}
 $d_{\mathds{B}^\mathsf{N}}$ is the Hamming distance. We will prove that 
 $d_{\mathds{S}_{\mathsf{N},\mathcal{P}}}$ is a distance
too, thus $d$ will also be a distance, being the sum of two distances.
 \begin{itemize}
\item Obviously, $d_{\mathds{S}_{\mathsf{N},\mathcal{P}}}(s,\check{s})\geqslant 0$, and if $s=\check{s}$, then 
$d_{\mathds{S}_{\mathsf{N},\mathcal{P}}}(s,\check{s})=0$. Conversely, if $d_{\mathds{S}_{\mathsf{N},\mathcal{P}}}(s,\check{s})=0$, then 
$\forall k \in \mathds{N}, v^k=\check{v}^k$ due to the 
definition of $d$. Then, as digits between positions $p+1$ and $p+n$ are null and correspond to $|u^0-\check{u}^0|$, we can conclude that $u^0=\check{u}^0$. An extension of this result to the whole first $n \times \max{(\mathcal{P})}$ blocs leads to $u^i=\check{u}^i$, $\forall i \leqslant v^0=\check{v}^0$, and by checking all the $n \times \max{(\mathcal{P})}$ blocs, $u=\check{u}$.
 \item $d_{\mathds{S}_{\mathsf{N},\mathcal{P}}}$ is clearly symmetric 
($d_{\mathds{S}_{\mathsf{N},\mathcal{P}}}(s,\check{s})=d_{\mathds{S}_{\mathsf{N},\mathcal{P}}}(\check{s},s)$). 
\item The triangle inequality is obtained because the absolute value satisfies it as well.
 \end{itemize}
\end{proof}

Before being able to study the topological behavior of the general 
chaotic iterations, we must first establish that:

\begin{prpstn}
 For all $f:\mathds{B}^\mathsf{N} \longrightarrow \mathds{B}^\mathsf{N} $, the function $G_f$ is continuous on 
$\left( \mathcal{X},d\right)$.
\end{prpstn}

\begin{proof}
We will show this result by using the sequential continuity. Consider a
sequence $x^n=(e^n,(u^n,v^n)) \in \mathcal{X}_{\mathsf{N},\mathcal{P}}^\mathds{N}$ such
that $d(x^n,x) \longrightarrow 0$, for some $x=(e,(u,v))\in
\mathcal{X}_{\mathsf{N},\mathcal{P}}$. We will show that
$d\left(G_f(x^n),G_f(x)\right) \longrightarrow 0$.
Remark that $u$ and $v$ are sequences of sequences.

As $d(x^n,x) \longrightarrow 0$, there exists 
$n_0\in\mathds{N}$ such that 
$d(x^n,x) < 10^{-(p+n \max{(\mathcal{P})})}$
(its $p+n \max{(\mathcal{P})}$ first digits are null). 
In particular, $\forall n \geqslant n_0, e^n=e$,
as the Hamming distance between the integral parts of
$x$ and $\check{x}$ is 0. Similarly, due to the nullity 
of the $p+n \max{(\mathcal{P})}$ first digits of 
$d(x^n,x)$, we can conclude that $\forall n \geqslant n_0$,
$(v^n)^0=v^0$, and that $\forall n \geqslant n_0$,
$(u^n)^0=u^0$, $(u^n)^1=u^1$, ..., $(u^n)^{v^0-1}=u^{v^0-1}$.
This implies that:
\begin{itemize}
\item $G_f(x^n)_1=G_f(x)_1$: they have the same
Boolean vector as first coordinate.
\item $d_{\mathds{S}_{\mathsf{N},\mathcal{P}}}(\Sigma (u^n,v^n); \Sigma(u,v)) = 10^{p+n \max{(\mathcal{P})}} d_{\mathds{S}_{\mathsf{N},\mathcal{P}}}((u^n,v^n); (u,v))$. As the right part of the equality tends
to 0, we can deduce that it is also the case for the left part of the equality, and so
$G_f(x^n)_2$ is convergent to $G_f(x)_2$.
\end{itemize}
\end{proof}

\subsection{$\Gamma_{\mathcal{P}}(f)$ as an extension of  $\Gamma(f)$}

Let $\mathcal{P}=\{p_1, p_2, \hdots, p_\mathsf{p}\}$.
We define the directed graph $\Gamma_{\mathcal{P}}(f)$ as follows.
\begin{itemize}
\item Its vertices are the $2^\mathsf{N}$ elements of $\mathds{B}^\mathsf{N}$.
\item Each vertex has $\displaystyle{\sum_{i=1}^\mathsf{p} \mathsf{N}^{p_i}}$ arrows, namely all the $p_1, p_2, \hdots, p_\mathsf{p}$ tuples 
  having their elements in $\llbracket 1, \mathsf{N} \rrbracket $.
\item There is an arc labeled $u_0, \hdots, u_{p_i-1}$, $i \in \llbracket 1, \mathsf{p} \rrbracket$ between vertices $x$ and $y$ if and only if 
$y=F_{f,p_i} (x, (u_0, \hdots, u_{p_i-1})) $.
\end{itemize}

It is not hard to see that the graph $\Gamma_{\{1\}}(f)$ is 
$\Gamma(f)$ formerly introduced in~\cite{bcgr11:ip} for the $\textit{CIPRNG}_f^1(u)$ generator,
which is indeed $\textit{CIPRNG}_f^2(u,(1)_{n \in \mathds{N}})$.

\begin{figure}[ht]
  \centering
  \subfigure[$\Gamma(f_0)$]{
      \begin{minipage}{0.45\textwidth}
        \centering
        \includegraphics[scale=0.85]{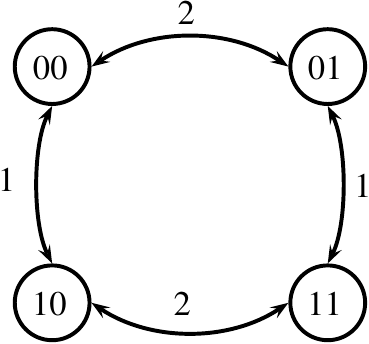}
      \end{minipage}
    \label{graphe1}
    }
    \subfigure[$\Gamma_{\{2,3\}}(f_0)$]{
      \begin{minipage}{0.3\textwidth}
        \centering
          \includegraphics[scale=0.85]{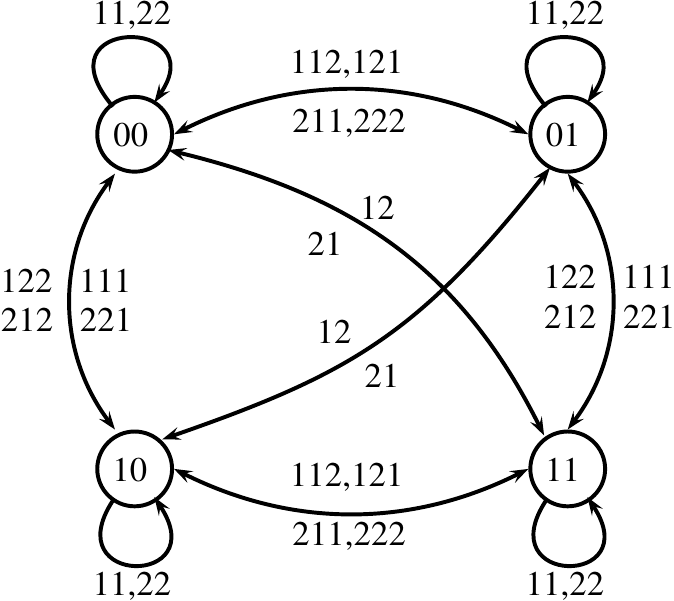}
      \end{minipage}
    \label{graphe2}
    }
  \caption{Iterating $f_0:(x_1,x_2) \mapsto (\overline{x_1}, \overline{x_2})$}
  \label{fig:itg}
\end{figure}

\begin{xpl}
Consider for instance $\mathsf{N}=2$, 
Let $f_0:\mathds{B}^2 \longrightarrow \mathds{B}^2$ be the negation function,
\textit{i.e.}, $f_0(x_1,x_2) = (\overline{x_1}, \overline{x_2})$, and consider
$\mathcal{P}=\{2,3\}$. The graphs of iterations are given in 
Figure~\ref{fig:itg}.
Figure~\ref{graphe1} shows what happens when each iteration result  
is displayed .
On the contrary, Figure~\ref{graphe2} illustrates what happens when  2 or 3 modifications 
are systematically applied 
before results are generated. 
Notice that here, the orientations of arcs are not necessary 
since the function $f_0$ is equal to its inverse $f_0^{-1}$. 
\end{xpl}

\subsection{Proofs of chaos}

We will show that,
\begin{prpstn}
\label{prop:trans}
 $\Gamma_{\mathcal{P}}(f)$ is strongly connected if and only if $G_f$ is 
topologically transitive on $(\mathcal{X}_{\mathsf{N},\mathcal{P}}, d)$.
\end{prpstn}

\begin{proof}
Suppose that $\Gamma_{\mathcal{P}}(f)$ is strongly connected. 
Let $x=(e,(u,v)),\check{x}=(\check{e},(\check{u},\check{v})) 
\in \mathcal{X}_{\mathsf{N},\mathcal{P}}$ and $\varepsilon >0$.
We will find a point $y$ in the open ball $\mathcal{B}(x,\varepsilon )$ and
$n_0 \in \mathds{N}$ such that $G_f^{n_0}(y)=\check{x}$: this strong transitivity
will imply the transitivity property.
We can suppose that $\varepsilon <1$ without loss of generality. 

Let us denote by $(E,(U,V))$  the elements of $y$. As
$y$ must be in $\mathcal{B}(x,\varepsilon)$ and  $\varepsilon < 1$,
$E$ must be equal to $e$. Let $k=\lfloor \log_{10} (\varepsilon) \rfloor +1$.
$d_{\mathds{S}_{\mathsf{N},\mathcal{P}}}((u,v),(U,V))$ must be lower than
$\varepsilon$, so the $k$ first digits of the fractional part of 
$d_{\mathds{S}_{\mathsf{N},\mathcal{P}}}((u,v),(U,V))$ are null.
Let $k_1$ be the smallest integer such that, if $V^0=v^0$, ...,  $V^{k_1}=v^{k_1}$,
 $U^0=u^0$, ..., $U^{\sum_{l=0}^{k_1}V^l-1} = u^{\sum_{l=0}^{k_1}v^l-1}$.
Then $d_{\mathds{S}_{\mathsf{N},\mathcal{P}}}((u,v),(U,V))<\varepsilon$.
In other words, any $y$ of the form $(e,((u^0, ..., u^{\sum_{l=0}^{k_1}v^l-1}),
(v^0, ..., v^{k_1}))$ is in $\mathcal{B}(x,\varepsilon)$.

Let $y^0$ such a point and $z=G_f^{k_1}(y^0) = (e',(u',v'))$. $\Gamma_{\mathcal{P}}(f)$
being strongly connected, there is a path between $e'$ and $\check{e}$. Denote
by $a_0, \hdots, a_{k_2}$ the edges visited by this path. We denote by
$V^{k_1}=|a_0|$ (number of terms in the finite sequence $a_1$),
$V^{k_1+1}=|a_1|$, ..., $V^{k_1+k_2}=|a_{k_2}|$, and by 
$U^{k_1}=a_0^0$, $U^{k_1+1}=a_0^1$, ..., $U^{k_1+V_{k_1}-1}=a_0^{V_{k_1}-1}$,
$U^{k_1+V_{k_1}}=a_1^{0}$, $U^{k_1+V_{k_1}+1}=a_1^{1}$,...

Let

\begin{eqnarray*}
y&=&(e,(
(u^0, \dots, u^{\sum_{l=0}^{k_1}v^l-1}, a_0^0, \dots, a_0^{|a_0|}, a_1^0, \dots, 
a_1^{|a_1|},\dots, a_{k_2}^0, \dots, a_{k_2}^{|a_{k_2}|},
 \check{u}^0, \check{u}^1, \dots), \\
&&\qquad(v^0, \dots, v^{k_1},|a_0|, \dots,
 |a_{k_2}|,\check{v}^0, \check{v}^1, \dots))).
\end{eqnarray*}
So $y\in \mathcal{B}(x,\varepsilon)$
 and $G_{f}^{k_1+k_2}(y)=\check{x}$.

Conversely, if $\Gamma_{\mathcal{P}}(f)$ is not strongly connected, then there are 
2 vertices $e_1$ and $e_2$ such that there is no path between $e_1$ and $e_2$.
Thus, it is impossible to find $(u,v)\in \mathds{S}_{\mathsf{N},\mathcal{P}}$
and $n\in \mathds{N}$ such that $G_f^n(e,(u,v))_1=e_2$. The open ball $\mathcal{B}(e_2, 1/2)$
cannot be reached from any neighborhood of $e_1$, and thus $G_f$ is not transitive.
\end{proof}

We now show  that,
\begin{prpstn}
If $\Gamma_{\mathcal{P}}(f)$ is strongly connected, then $G_f$ is 
regular on $(\mathcal{X}_{\mathsf{N},\mathcal{P}}, d)$.
\end{prpstn}

\begin{proof}
Let $x=(e,(u,v)) \in \mathcal{X}_{\mathsf{N},\mathcal{P}}$ and $\varepsilon >0$. 
As in the proofs of Prop.~\ref{prop:trans}, let $k_1 \in \mathds{N}$ such
that 
$$\left\{(e, ((u^0, \dots, u^{v^{k_1-1}},U^0, U^1, \dots),(v^0, \dots, v^{k_1},V^0, V^1, \dots)) \mid \right.$$
$$\left.\forall i,j \in \mathds{N}, U^i \in \llbracket 1, \mathsf{N} \rrbracket, V^j \in \mathcal{P}\right\}
\subset \mathcal{B}(x,\varepsilon),$$
and $y=G_f^{k_1}(e,(u,v))$. $\Gamma_{\mathcal{P}}(f)$ being strongly connected,
there is at least a path from the Boolean state $y_1$ of $y$ to $e$.
Denote by $a_0, \hdots, a_{k_2}$ the edges of such a path.
Then the point:\linebreak
$(e,((u^0, \dots, u^{v^{k_1-1}},a_0^0, \dots, a_0^{|a_0|}, a_1^0, \dots, a_1^{|a_1|},\dots, 
 a_{k_2}^0, \dots,$ \linebreak 
$\,a_{k_2}^{|a_{k_2}|},u^0, \dots, u^{v^{k_1-1}},a_0^0, \dots,a_{k_2}^{|a_{k_2}|}\dots),$\linebreak
$(v^0, \dots, v^{k_1}, |a_0|, \dots, |a_{k_2}|,v^0, \dots, v^{k_1}, |a_0|, \dots, |a_{k_2}|,\dots))$
is a periodic point in the neighborhood $\mathcal{B}(x,\varepsilon)$ of $x$.
\end{proof}

$G_f$ being topologically transitive and regular, we can thus conclude that
\begin{thrm}
Function $G_f$ is chaotic on $(\mathcal{X}_{\mathsf{N},\mathcal{P}},d)$ if
and only if its iteration graph $\Gamma_{\mathcal{P}}(f)$ is strongly connected.
\end{thrm}

\begin{crllr}
  The pseudorandom number generator $\chi_{\textit{14Secrypt}}$ is not chaotic
  on $(\mathcal{X}_{\mathsf{N},\{b\}},d)$ for the negation function.
\end{crllr}
\begin{proof}
  In this context, $\mathcal{P}$ is the singleton $\{b\}$.
  If $b$ is even, no vertex $e$ of $\Gamma_{\{b\}}(f_0)$ can reach 
  its neighborhood and thus $\Gamma_{\{b\}}(f_0)$ is not strongly connected. 
  If $b$ is odd, no vertex $e$ of $\Gamma_{\{b\}}(f_0)$ can reach itself 
  and thus $\Gamma_{\{b\}}(f_0)$ is not strongly connected.
\end{proof}

\subsection{Comparison with other well-known generators}

\begin{table}
\centering
  \begin{tabular}{c|ccccccc}
  PRNG & LCG & MRG & AWC & SWB & SWC & GFSR & INV\\
  \hline
  NIST & 11 & 14 & 15 & 15 & 14 & 14 & 14\\
  DieHARD & 16 & 16 & 15 & 16 &18 & 16 & 16
  \end{tabular}
  \caption{Statistical evaluation of known PRNGs: number of succeeded tests}
  \label{table:comparisonWithout}
\end{table}

\begin{table}
\centering
  \begin{tabular}{c|ccccccc}
  PRNG & LCG & MRG & AWC & SWB & SWC & GFSR & INV\\
  \hline
  NIST & 15 & 15 & 15 & 15 & 15 & 15 & 15\\
  DieHARD & 18 & 18 & 18 & 18 & 18 & 18 & 18
  \end{tabular}
  \caption{Statistical effects of CIPRNG on the succeeded tests}
  \label{table:comparisonWith}
\end{table}
The objective of this section is to evaluate the statistical performance of the 
proposed CIPRNG method, by comparing the effects of its application on well-known
but defective generators. We considered during the experiments the following PRNGs:
linear congruential generator (LCG), multiple recursive generators (MRG)
add-with-carry (AWC), subtract-with-borrow (SWB), shift-with-carry (SWC)
Generalized Feedback Shift Register (GFSR), and nonlinear inversive generator.
A general overview and a reminder of these  generators 
can be found, for instance, in the documentation of the TestU01 statistical
battery of tests~\cite{LEcuyerS07}. For each studied generator, we have compared
their scores according to both NIST~\cite{Nist10} and DieHARD~\cite{Marsaglia1996}
statistical batteries of tests, by launching them alone or inside the $\textit{CIPRNG}_f^2(v,v)$
dynamical system, where $v$ is the considered PRNG set with most usual parameters,
and $f$ is the vectorial negation. 

Obtained results are reproduced in Tables~\ref{table:comparisonWithout} and \ref{table:comparisonWith}.
As can be seen, all these generators considered alone failed to pass either the 15 NIST tests or the
18 DieHARD ones, while both batteries of tests are always passed when applying the $\textit{CIPRNG}_f^2$
post-treatment. Other results in the same direction, which can be found in~\cite{bfgw11:ip}, illustrate
the fact that operating a provable chaotic post-treatment on defective generators tends to improve
their statistical profile. 

Such post-treatment depending on the properties of the inputted function $f$, we need to recall a general scheme to produce
functions and an iteration number $b$ such that $\Gamma_{\{b\}}$ is strongly connected.

\section{Functions with Strongly Connected $\Gamma_{\{b\}}(f)$}\label{sec:SCCfunc}
First of all, let $f: \Bool^{{\mathsf{N}}} \rightarrow \Bool^{{\mathsf{N}}}$.
It has been shown~\cite[Theorem 4]{bcgr11:ip} that 
if its iteration graph $\Gamma(f)$ is strongly connected, then 
the output of $\chi_{\textit{14Secrypt}}$ follows 
a law that tends to the uniform distribution 
if and only if its Markov matrix is a doubly stochastic one.
In~\cite[Section 4]{DBLP:conf/secrypt/CouchotHGWB14},
we have presented a general scheme which generates 
function with strongly connected iteration graph $\Gamma(f)$ and
with doubly stochastic Markov probability matrix.

Basically, let us consider the ${\mathsf{N}}$-cube. Let us next 
remove one Hamiltonian cycle in this one. When an edge $(x,y)$ 
is removed, an edge $(x,x)$ is added. 

\begin{xpl}
For instance, the iteration graph $\Gamma(f^*)$ 
(given in Figure~\ref{fig:iteration:f*}) 
is the $3$-cube in which the Hamiltonian cycle 
$000,100,101,001,011,111,$ $110,010,000$ 
has been removed.
\end{xpl}

We  have first proven the following result, which 
states that the ${\mathsf{N}}$-cube without one
Hamiltonian cycle 
has the awaited property with regard to the connectivity.

\begin{thrm}
The iteration graph $\Gamma(f)$ issued from
the ${\mathsf{N}}$-cube where an Hamiltonian 
cycle is removed, is strongly connected.
\end{thrm}

Moreover, when all the transitions have the same probability ($\frac{1}{n}$),
we have proven the following results:
\begin{thrm}
The Markov Matrix $M$ resulting from the ${\mathsf{N}}$-cube in
which an Hamiltonian 
cycle is removed, is doubly stochastic.
\end{thrm}

Let us consider now a ${\mathsf{N}}$-cube where an Hamiltonian 
cycle is removed.
Let $f$ be the corresponding function.
The question which remains to be solved is:
\emph{can we always find $b$ such that $\Gamma_{\{b\}}(f)$ is strongly connected?}

The answer is indeed positive. Furthermore, we have the following results which are stronger
than previous ones. 
\begin{thrm}
There exists $b \in \Nats$ such that $\Gamma_{\{b\}}(f)$ is complete.
\end{thrm}
\begin{proof}
There is an arc $(x,y)$ in the 
graph $\Gamma_{\{b\}}(f)$ if and only if $M^b_{xy}$ is positive
where $M$ is the Markov matrix of $\Gamma(f)$.
It has been shown in~\cite[Lemma 3]{bcgr11:ip}  that $M$ is regular.
Thus, there exists $b$ such that there is an arc between any $x$ and $y$.
\end{proof}

This section ends with the idea of removing a Hamiltonian cycle in the 
$\mathsf{N}$-cube. 
In such a context, the Hamiltonian cycle is equivalent to a Gray code.
Many approaches have been proposed as a way to build such codes, for instance 
the Reflected Binary Code. In this one and 
for a $\mathsf{N}$-length cycle, one of the bits is exactly switched 
$2^{\mathsf{N}-1}$ times whereas the other bits are modified at most 
$\left\lfloor \dfrac{2^{\mathsf{N-1}}}{\mathsf{N}-1} \right\rfloor$ times.
It is clear that the function that is built from such a code would
not provide a uniform output.  

The next section presents how to build balanced Hamiltonian cycles in the 
$\mathsf{N}$-cube with the objective to embed them into the 
pseudorandom number generator.

\section{Balanced Hamiltonian Cycle}\label{sec:hamilton}
Many approaches have been developed to solve the problem of building
a Gray code in a $\mathsf{N}$-cube~\cite{Robinson:1981:CS,DBLP:journals/combinatorics/BhatS96,ZanSup04,Bykov2016}, according to properties 
the produced code has to verify.
For instance,~\cite{DBLP:journals/combinatorics/BhatS96,ZanSup04} focus on
balanced Gray codes. In the transition sequence of these codes, 
the number of transitions of each element must differ
at most by 2.
This uniformity is a global property on the cycle, \textit{i.e.},
a property that is established while traversing the whole cycle.
On the other hand, when the objective is to follow a subpart 
of the Gray code and to switch each element approximately the 
same amount of times,
local properties are wished.
For instance, the locally balanced property is studied in~\cite{Bykov2016} 
and an algorithm that establishes locally balanced Gray codes is given.
 
The current context is to provide a function 
$f:\Bool^{\mathsf{N}} \rightarrow \Bool^{\mathsf{N}}$ by removing an Hamiltonian 
cycle in the $\mathsf{N}$-cube. Such a function is going to be iterated
$b$ times to produce a pseudorandom number,
\textit{i.e.}, a vertex in the 
$\mathsf{N}$-cube.
Obviously, the number of iterations $b$ has to be sufficiently large 
to provide a uniform output distribution.
To reduce the number of iterations, it can be claimed
that the provided Gray code
should ideally possess both balanced and locally balanced properties.
However, both algorithms are incompatible with the second one:
balanced Gray codes that are generated by state of the art works~\cite{ZanSup04,DBLP:journals/combinatorics/BhatS96} are not locally balanced. Conversely,
locally balanced Gray codes yielded by Igor Bykov approach~\cite{Bykov2016}
are not globally balanced.
This section thus shows how the non deterministic approach 
presented in~\cite{ZanSup04} has been automatized to provide balanced 
Hamiltonian paths such that, for each subpart, 
the number of switches of each element is as uniform as possible.

\subsection{Analysis of the Robinson-Cohn extension algorithm}
As far as we know three works, 
namely~\cite{Robinson:1981:CS},~\cite{DBLP:journals/combinatorics/BhatS96}, 
and~\cite{ZanSup04} have addressed the problem of providing an approach
to produce balanced gray code.
The authors of~\cite{Robinson:1981:CS} introduced an inductive approach
aiming at producing balanced Gray codes, assuming the user gives 
a special subsequence of the transition sequence at each induction step.
This work has been strengthened in~\cite{DBLP:journals/combinatorics/BhatS96}
where the authors have explicitly shown how to build such a subsequence.
Finally the authors of~\cite{ZanSup04} have presented 
the \emph{Robinson-Cohn extension} 
algorithm. Their rigorous presentation of this algorithm
has mainly allowed them to prove two properties.
The former states that if 
$\mathsf{N}$ is a 2-power, a balanced Gray code is always totally balanced.
The latter states that for every $\mathsf{N}$ there 
exists a Gray code such that all transition count numbers 
are 2-powers whose exponents are either equal
or differ from each other by 1.
However, the authors do not prove that the approach allows to build 
(totally balanced) Gray codes. 
What follows shows that this fact is established and first recalls the approach.

Let be given a $\mathsf{N}-2$-bit Gray code whose transition sequence is
$S_{\mathsf{N}-2}$. What follows is the 
 \emph{Robinson-Cohn extension} method~\cite{ZanSup04}
which produces a $\mathsf{N}$-bits Gray code.

\begin{enumerate}
\item \label{item:nondet}Let $l$ be an even positive integer. Find 
$u_1, u_2, \dots , u_{l-2}, v$ (maybe empty) subsequences of $S_{\mathsf{N}-2}$ 
such that $S_{\mathsf{N}-2}$ is the concatenation of 
$$
s_{i_1}, u_0, s_{i_2}, u_1, s_{i_3}, u_2, \dots , s_{i_l-1}, u_{l-2}, s_{i_l}, v
$$
where $i_1 = 1$, $i_2 = 2$, and $u_0 = \emptyset$ (the empty sequence).
\item\label{item:u'}Replace in $S_{\mathsf{N}-2}$ the sequences $u_0, u_1, u_2, \ldots, u_{l-2}$ 
  by 
  $\mathsf{N} - 1,  u'(u_1,\mathsf{N} - 1, \mathsf{N}) , u'(u_2,\mathsf{N}, \mathsf{N} - 1), u'(u_3,\mathsf{N} - 1,\mathsf{N}), \dots, u'(u_{l-2},\mathsf{N}, \mathsf{N} - 1)$
  respectively, where $u'(u,x,y)$ is the sequence $u,x,u^R,y,u$ such that 
  $u^R$ is $u$ in reversed order. 
  The obtained sequence is further denoted as $U$.
\item\label{item:VW} Construct the sequences $V=v^R,\mathsf{N},v$, $W=\mathsf{N}-1,S_{\mathsf{N}-2},\mathsf{N}$, and let $W'$ be $W$ where the first 
two elements have been exchanged.
\item The transition sequence $S_{\mathsf{N}}$ is thus the concatenation $U^R, V, W'$.
\end{enumerate} 

It has been proven in~\cite{ZanSup04} that 
$S_{\mathsf{N}}$ is the transition sequence of a cyclic $\mathsf{N}$-bits Gray code 
if $S_{\mathsf{N}-2}$ is. 
However, step~(\ref{item:nondet}) is not a constructive 
step that precises how to select the subsequences which ensure that 
yielded Gray code is balanced.
Following sections show how to choose the sequence $l$ to have the balance property.

\subsection{Balanced Codes}
Let us first recall how to formalize the balance property of a Gray code.
Let  $L = w_1, w_2, \dots, w_{2^\mathsf{N}}$ be the sequence 
of a $\mathsf{N}$-bits cyclic Gray code. 
The  transition sequence 
$S = s_1, s_2, \dots, s_{2^n}$, $s_i$, $1 \le i \le 2^\mathsf{N}$,
indicates which bit position changes between 
codewords at index $i$ and $i+1$ modulo $2^\mathsf{N}$. 
The \emph{transition count} function  
$\textit{TC}_{\mathsf{N}} : \{1,\dots, \mathsf{N}\} \rightarrow \{0, \ldots, 2^{\mathsf{N}}\}$ 
gives the number of times $i$ occurs in $S$,
\textit{i.e.}, the number of times 
the bit $i$ has been switched in $L$.   

The Gray code is \emph{totally balanced} if $\textit{TC}_{\mathsf{N}}$ 
is constant (and equal to $\frac{2^{\mathsf{N}}}{\mathsf{N}}$).
It is \emph{balanced} if for any two bit indices $i$ and $j$, 
$|\textit{TC}_{\mathsf{N}}(i) - \textit{TC}_{\mathsf{N}}(j)| \le  2$.

\begin{xpl}
Let $L^*=000,100,101,001,011,111,$ $110,010$ be the Gray code that corresponds to 
the Hamiltonian cycle that has been removed in $f^*$.
Its transition sequence is $S=3,1,3,2,3,1,3,2$ and its transition count function is 
$\textit{TC}_3(1)= \textit{TC}_3(2)=2$ and  $\textit{TC}_3(3)=4$. Such a Gray code is balanced. 

Let  $L^4$ $=0000,0010,0110,1110,1111,0111,0011,0001,0101,0100,1100,1101,1001,1011,1010,1000$
be a cyclic Gray code. Since $S=2,3,4,1,4,$ $3,2,3,1,4,1,3,2,1,2,4$, $\textit{TC}_4$ is equal to 4 everywhere, this code
is thus totally balanced.

On the contrary, for the standard $4$-bits Gray code  
$L^{\textit{st}}=0000,0001,0011, 
0010,0110,0111,0101,0100,$ \newline $1100,1101,1111,1110,1010,1011,1001,1000$,
we have $\textit{TC}_4(1)=8$ $\textit{TC}_4(2)=4$ $\textit{TC}_4(3)=\textit{TC}_4(4)=2$ and
the code is neither balanced nor totally balanced.
\end{xpl}

\begin{thrm}\label{prop:balanced}
Let $\mathsf{N}$ in $\Nats^*$, and $a_{\mathsf{N}}$ be defined by
$a_{\mathsf{N}}= 2 \left\lfloor \dfrac{2^{\mathsf{N}}}{2\mathsf{N}} \right\rfloor$. 
There exists then a sequence $l$ in 
step~(\ref{item:nondet}) of the \emph{Robinson-Cohn extension} algorithm
such that all the transition counts $\textit{TC}_{\mathsf{N}}(i)$ 
are $a_{\mathsf{N}}$ or $a_{\mathsf{N}}+2$ 
for any $i$, $1 \le i \le \mathsf{N}$.
\end{thrm}

The proof is done by induction on $\mathsf{N}$. Let us immediately verify 
that it is established for both odd and even smallest values, \textit{i.e.}, 
$3$ and $4$.
For the initial case where $\mathsf{N}=3$, \textit{i.e.}, $\mathsf{N-2}=1$ we successively have:  $S_1=1,1$, $l=2$,  $u_0 = \emptyset$, and $v=\emptyset$.
Thus again the algorithm successively produces 
$U= 1,2,1$, $V = 3$, $W= 2, 1, 1,3$, and $W' = 1,2,1,3$.
Finally, $S_3$ is $1,2,1,3,1,2,1,3$ which obviously verifies the theorem.    
 For the initial case where $\mathsf{N}=4$, \textit{i.e.}, $\mathsf{N-2}=2$ 
we successively have:  $S_1=1,2,1,2$, $l=4$, 
$u_0,u_1,u_2 = \emptyset,\emptyset,\emptyset$, and $v=\emptyset$.
Thus again the algorithm successively produces 
$U= 1,3,2,3,4,1,4,3,2$, $V = 4$, $W= 3, 1, 2, 1,2, 4$, and $W' = 1, 3, 2, 1,2, 4 $.
Finally, $S_4$ is 
$
2,3,4,1,4,3,2,3,1,4,1,3,2,1,2,4
$ 
such that $\textit{TC}_4(i) = 4$ and the theorem is established for 
odd and even initial values.

For the inductive case, let us first define some variables.
Let $c_{\mathsf{N}}$ (resp. $d_{\mathsf{N}}$) be the number of elements 
whose transition count is exactly $a_{\mathsf{N}}$ (resp $a_{\mathsf{N}} +2$).
Both of these variables are defined by the system 
 
\[
\left\{
\begin{array}{lcl}
c_{\mathsf{N}} + d_{\mathsf{N}} & = & \mathsf{N} \\
c_{\mathsf{N}}a_{\mathsf{N}} + d_{\mathsf{N}}(a_{\mathsf{N}}+2) & = & 2^{\mathsf{N}} 
\end{array}
\right.
\Leftrightarrow 
\left\{
\begin{array}{lcl}
d_{\mathsf{N}} & = & \dfrac{2^{\mathsf{N}} -\mathsf{N}.a_{\mathsf{N}}}{2} \\
c_{\mathsf{N}} &= &\mathsf{N} -  d_{\mathsf{N}}
\end{array}
\right.
\]

Since $a_{\mathsf{N}}$ is even, $d_{\mathsf{N}}$ is an integer.
Let us first prove that both $c_{\mathsf{N}}$ and  $d_{\mathsf{N}}$ are positive
integers.
Let $q_{\mathsf{N}}$ and $r_{\mathsf{N}}$, respectively, be  
the quotient and the remainder in the Euclidean division
of $2^{\mathsf{N}}$ by $2\mathsf{N}$, \textit{i.e.}, 
$2^{\mathsf{N}} = q_{\mathsf{N}}.2\mathsf{N} + r_{\mathsf{N}}$, with $0 \le r_{\mathsf{N}} <2\mathsf{N}$.
First of all, the integer $r$ is even since $r_{\mathsf{N}} = 2^{\mathsf{N}} - q_{\mathsf{N}}.2\mathsf{N}= 2(2^{\mathsf{N}-1} - q_{\mathsf{N}}.\mathsf{N})$. 
Next,  $a_{\mathsf{N}}$ is $\frac{2^{\mathsf{N}}-r_{\mathsf{N}}}{\mathsf{N}}$. Consequently 
$d_{\mathsf{N}}$ is $r_{\mathsf{N}}/2$ and is thus a positive integer s.t. 
$0 \le d_{\mathsf{N}} <\mathsf{N}$.
The proof for $c_{\mathsf{N}}$ is obvious.

For any $i$, $1 \le i \le \mathsf{N}$, let $zi_{\mathsf{N}}$ (resp. $ti_{\mathsf{N}}$ and $bi_{\mathsf{N}}$) 
be the occurrence number of element $i$ in the sequence $u_0, \dots, u_{l-2}$ 
(resp. in the sequences $s_{i_1}, \dots , s_{i_l}$ and $v$)
in step (\ref{item:nondet}) of the algorithm.

Due to the definition of $u'$ in step~(\ref{item:u'}), 
$3.zi_{\mathsf{N}} + ti_{\mathsf{N}}$ is the
 number of element $i$ in the sequence $U$.
It is clear that the number of element $i$ in the sequence $V$ is 
$2bi_{\mathsf{N}}$ due to step (\ref{item:VW}). 
We thus have the following system:
$$
\left\{
\begin{array}{lcl}
3.zi_{\mathsf{N}} + ti_{\mathsf{N}} + 2.bi_{\mathsf{N}} + \textit{TC}_{\mathsf{N}-2}(i) &= &\textit{TC}_{\mathsf{N}}(i) \\
zi_{\mathsf{N}} + ti_{\mathsf{N}}  + bi_{\mathsf{N}} & =& \textit{TC}_{\mathsf{N}-2}(i)
\end{array}
\right.
\qquad 
\Leftrightarrow 
$$

\begin{equation}
\left\{
\begin{array}{lcl}
zi_{\mathsf{N}} &= & 
\dfrac{\textit{TC}_{\mathsf{N}}(i) - 2.\textit{TC}_{\mathsf{N}-2}(i) - bi_{\mathsf{N}}}{2}\\
ti_{\mathsf{N}} &= & \textit{TC}_{\mathsf{N}-2}(i)-zi_{\mathsf{N}}-bi_{\mathsf{N}}
\end{array}
\right.
\label{eq:sys:zt1}
\end{equation}

In this set of 2 equations with 3 unknown variables, let $b_i$ be set with 0.
In this case, since $\textit{TC}_{\mathsf{N}}$ is even (equal to $a_{\mathsf{N}}$ 
or to $a_{\mathsf{N}}+2$), the variable  $zi_{\mathsf{N}}$ is thus an integer.
Let us now prove that the resulting system has always positive integer 
solutions $z_i$, $t_i$, $0 \le z_i, t_i \le \textit{TC}_{\mathsf{N}-2}(i)$ 
and s.t. their sum is equal to $\textit{TC}_{\mathsf{N}-2}(i)$. 
This latter constraint is obviously established if the system has a solution.
We thus have the following system.

\begin{equation}
\left\{
\begin{array}{lcl}
zi_{\mathsf{N}} &= & 
\dfrac{\textit{TC}_{\mathsf{N}}(i) - 2.\textit{TC}_{\mathsf{N}-2}(i) }{2}\\
ti_{\mathsf{N}} &= & \textit{TC}_{\mathsf{N}-2}(i)-zi_{\mathsf{N}}
\end{array}
\right.
\label{eq:sys:zt2}
\end{equation}

The definition of $\textit{TC}_{\mathsf{N}}(i)$ depends on the value of $\mathsf{N}$.
When $3 \le N \le 7$, values are defined as follows:
\begin{eqnarray*}
\textit{TC}_{3} & = & [2,2,4] \\
\textit{TC}_{5} & = & [6,6,8,6,6] \\
\textit{TC}_{7} & = & [18,18,20,18,18,18,18] \\
\\
\textit{TC}_{4} & = & [4,4,4,4] \\
\textit{TC}_{6} & = & [10,10,10,10,12,12] \\
\end{eqnarray*}
It is not difficult to check that all these instanciations verify the aforementioned constraints. 

When  $N  \ge 8$, $\textit{TC}_{\mathsf{N}}(i)$ is defined as follows:
\begin{equation}
\textit{TC}_{\mathsf{N}}(i) = \left\{
\begin{array}{l} 
a_{\mathsf{N}} \textrm{ if } 1 \le i \le c_{\mathsf{N}} \\
a_{\mathsf{N}}+2 \textrm{ if } c_{\mathsf{N}} +1 \le i \le c_{\mathsf{N}} + d_{\mathsf{N}} 
\end{array}
\right.
\label{eq:TCN:def}
\end{equation}

We thus  have
\[
\begin{array}{rcl} 
\textit{TC}_{\mathsf{N}}(i) - 2.\textit{TC}_{\mathsf{N}-2}(i) 
&\ge& 
a_{\mathsf{N}} - 2(a_{\mathsf{N}-2}+2) \\
&\ge& 
\frac{2^{\mathsf{N}}-r_{\mathsf{N}}}{\mathsf{N}}
-2 \left( \frac{2^{\mathsf{N-2}}-r_{\mathsf{N-2}}}{\mathsf{N-2}}+2\right)\\
&\ge& 
\frac{2^{\mathsf{N}}-2N}{\mathsf{N}}
-2 \left( \frac{2^{\mathsf{N-2}}}{\mathsf{N-2}}+2\right)\\
&\ge& 
\frac{(\mathsf{N} -2).2^{\mathsf{N}}-2N.2^{\mathsf{N-2}}-6N(N-2)}{\mathsf{N.(N-2)}}\\
\end{array}
\]

A simple variation study of the function $t:\R \rightarrow \R$ such that 
$x \mapsto t(x) = (x -2).2^{x}-2x.2^{x-2}-6x(x-2)$ shows that 
its derivative is strictly positive if $x \ge 6$ and $t(8)=224$.
The integer $\textit{TC}_{\mathsf{N}}(i) - 2.\textit{TC}_{\mathsf{N}-2}(i)$ is thus positive 
for any $\mathsf{N} \ge 8$ and the proof is established.


For each element $i$, we are then left to choose $zi_{\mathsf{N}}$ positions 
among $\textit{TC}_{\mathsf{N}}(i)$, which leads to 
${\textit{TC}_{\mathsf{N}}(i) \choose zi_{\mathsf{N}} }$ possibilities.
Notice that all such choices lead to an Hamiltonian path.


\section{Mixing Time}\label{sec:hypercube}
This section considers functions $f: \Bool^{\mathsf{N}} \rightarrow \Bool^{\mathsf{N}} $ 
issued from an hypercube where an Hamiltonian path has been removed
as described in the previous section.
Notice that the iteration graph is always a subgraph of 
${\mathsf{N}}$-cube augmented with all the self-loop, \textit{i.e.}, all the 
edges $(v,v)$ for any $v \in \Bool^{\mathsf{N}}$. 
Next, if we add probabilities on the transition graph, iterations can be 
interpreted as Markov chains.

\begin{xpl}
Let us consider for instance  
the graph $\Gamma(f)$ defined 
in Figure~\ref{fig:iteration:f*} and 
the probability function $p$ defined on the set of edges as follows:
$$
p(e) \left\{
\begin{array}{ll}
= \frac{2}{3} \textrm{ if $e=(v,v)$ with $v \in \Bool^3$,}\\
= \frac{1}{6} \textrm{ otherwise.}
\end{array}
\right.  
$$
The matrix $P$ of the Markov chain associated to the function $f^*$ and to its probability function $p$ is 
\[
P=\dfrac{1}{6} \left(
\begin{array}{llllllll}
4&1&1&0&0&0&0&0 \\
1&4&0&0&0&1&0&0 \\
0&0&4&1&0&0&1&0 \\
0&1&1&4&0&0&0&0 \\
1&0&0&0&4&0&1&0 \\
0&0&0&0&1&4&0&1 \\
0&0&0&0&1&0&4&1 \\
0&0&0&1&0&1&0&4 
\end{array}
\right).
\]
\end{xpl}

A specific random walk in this modified hypercube is first 
introduced (see Section~\ref{sub:stop:formal}). We further 
 study this random walk in a theoretical way to 
provide an upper bound of fair sequences 
(see Section~\ref{sub:stop:bound}).
We finally complete this study with experimental
results that reduce this bound (Sec.~\ref{sub:stop:exp}).
For a general reference on Markov chains,
see~\cite{LevinPeresWilmer2006}, 
and particularly Chapter~5 on stopping times.

\subsection{Formalizing the Random Walk}\label{sub:stop:formal}

First of all, let $\pi$, $\mu$ be two distributions on $\Bool^{\mathsf{N}}$. The total
variation distance between $\pi$ and $\mu$ is denoted $\tv{\pi-\mu}$ and is
defined by
$$\tv{\pi-\mu}=\max_{A\subset \Bool^{\mathsf{N}}} |\pi(A)-\mu(A)|.$$ It is known that
$$\tv{\pi-\mu}=\frac{1}{2}\sum_{X\in\Bool^{\mathsf{N}}}|\pi(X)-\mu(X)|.$$ Moreover, if
$\nu$ is a distribution on $\Bool^{\mathsf{N}}$, one has
$$\tv{\pi-\mu}\leq \tv{\pi-\nu}+\tv{\nu-\mu}$$

Let $P$ be the matrix of a Markov chain on $\Bool^{\mathsf{N}}$. For any
$X\in \Bool^{\mathsf{N}}$, let $P(X,\cdot)$ be the distribution induced by the
${\rm bin}(X)$-th row of $P$, where ${\rm bin}(X)$ is the integer whose
binary encoding is $X$. If the Markov chain induced by $P$ has a stationary
distribution $\pi$, then we define
$$d(t)=\max_{X\in\Bool^{\mathsf{N}}}\tv{P^t(X,\cdot)-\pi}.$$

and

$$t_{\rm mix}(\varepsilon)=\min\{t \mid d(t)\leq \varepsilon\}.$$


Intuitively speaking,  $t_{\rm mix}(\varepsilon)$ is the time/steps required
to be sure to be $\varepsilon$-close to the stationary distribution, wherever
the chain starts.

One can prove that

$$t_{\rm mix}(\varepsilon)\leq \lceil\log_2(\varepsilon^{-1})\rceil t_{\rm mix}(\frac{1}{4})$$



Let $(X_t)_{t\in \mathbb{N}}$ be a sequence of $\Bool^{\mathsf{N}}$ valued random
variables. A $\mathbb{N}$-valued random variable $\tau$ is a {\it stopping
  time} for the sequence $(X_i)$ if for each $t$ there exists $B_t\subseteq
(\Bool^{\mathsf{N}})^{t+1}$ such that $\{\tau=t\}=\{(X_0,X_1,\ldots,X_t)\in B_t\}$. 
In other words, the event $\{\tau = t \}$ only depends on the values of 
$(X_0,X_1,\ldots,X_t)$, not on $X_k$ with $k > t$.

Let $(X_t)_{t\in \mathbb{N}}$ be a Markov chain and $f(X_{t-1},Z_t)$ a
random mapping representation of the Markov chain. A {\it randomized
  stopping time} for the Markov chain is a stopping time for
$(Z_t)_{t\in\mathbb{N}}$. If the Markov chain is irreducible and has $\pi$
as stationary distribution, then a {\it stationary time} $\tau$ is a
randomized stopping time (possibly depending on the starting position $X$),
such that  the distribution of $X_\tau$ is $\pi$:
$$\P_X(X_\tau=Y)=\pi(Y).$$

\subsection{Upper bound of Stopping Time}\label{sub:stop:bound}

A stopping time $\tau$ is a {\emph strong stationary time} if $X_{\tau}$ is
independent of $\tau$. The following result will be useful~\cite[Proposition~6.10]{LevinPeresWilmer2006},

\begin{thrm}\label{thm-sst}
If $\tau$ is a strong stationary time, then $d(t)\leq \max_{X\in\Bool^{\mathsf{N}}}
\P_X(\tau > t)$.
\end{thrm}

Let $E=\{(X,Y)\mid
X\in \Bool^{\mathsf{N}}, Y\in \Bool^{\mathsf{N}},\ X=Y \text{ or } X\oplus Y \in 0^*10^*\}$.
In other words, $E$ is the set of all the edges in the classical 
${\mathsf{N}}$-cube. 
Let $h$ be a function from $\Bool^{\mathsf{N}}$ into $\llbracket 1, {\mathsf{N}} \rrbracket$.
Intuitively speaking $h$ aims at memorizing for each node 
$X \in \Bool^{\mathsf{N}}$ whose edge is removed in the Hamiltonian cycle,
\textit{i.e.}, which bit in $\llbracket 1, {\mathsf{N}} \rrbracket$ 
cannot be switched.

We denote by $E_h$ the set $E\setminus\{(X,Y)\mid X\oplus Y =
0^{{\mathsf{N}}-h(X)}10^{h(X)-1}\}$. This is the set of the modified hypercube, 
\textit{i.e.}, the ${\mathsf{N}}$-cube where the Hamiltonian cycle $h$ 
has been removed.

We define the Markov matrix $P_h$ for each line $X$ and 
each column $Y$  as follows:
\begin{equation}
\left\{
\begin{array}{ll}
P_h(X,X)=\frac{1}{2}+\frac{1}{2{\mathsf{N}}} & \\
P_h(X,Y)=0 & \textrm{if  $(X,Y)\notin E_h$}\\
P_h(X,Y)=\frac{1}{2{\mathsf{N}}} & \textrm{if $X\neq Y$ and $(X,Y) \in E_h$}
\end{array}
\right.
\label{eq:Markov:rairo}
\end{equation} 

We denote by $\ov{h} : \Bool^{\mathsf{N}} \rightarrow \Bool^{\mathsf{N}}$ the function 
such that for any $X \in \Bool^{\mathsf{N}} $, 
$(X,\ov{h}(X)) \in E$ and $X\oplus\ov{h}(X)=0^{{\mathsf{N}}-h(X)}10^{h(X)-1}$. 
The function $\ov{h}$ is said to be {\it square-free} if for every $X\in \Bool^{\mathsf{N}}$,
$\ov{h}(\ov{h}(X))\neq X$. 

\begin{lmm}\label{lm:h}
If $\ov{h}$ is bijective and square-free, then $h(\ov{h}^{-1}(X))\neq h(X)$.
\end{lmm}

\begin{proof}
Let $\ov{h}$ be bijective.
Let $k\in \llbracket 1, {\mathsf{N}} \rrbracket$ s.t. $h(\ov{h}^{-1}(X))=k$.
Then $(\ov{h}^{-1}(X),X)$ belongs to $E$ and 
$\ov{h}^{-1}(X)\oplus X = 0^{{\mathsf{N}}-k}10^{k-1}$.
Let us suppose $h(X) = h(\ov{h}^{-1}(X))$. In such a case, $h(X) =k$.
By definition of $\ov{h}$, $(X, \ov{h}(X)) \in E $ and 
$X\oplus\ov{h}(X)=0^{{\mathsf{N}}-h(X)}10^{h(X)-1} = 0^{{\mathsf{N}}-k}10^{k-1}$.
Thus $\ov{h}(X)= \ov{h}^{-1}(X)$, which leads to $\ov{h}(\ov{h}(X))= X$.
This contradicts the square-freeness of $\ov{h}$.
\end{proof}

Let $Z$ be a random variable that is uniformly distributed over
$\llbracket 1, {\mathsf{N}} \rrbracket \times \Bool$.
For $X\in \Bool^{\mathsf{N}}$, we
define, with $Z=(i,b)$,  
$$
\left\{
\begin{array}{ll}
f(X,Z)=X\oplus (0^{{\mathsf{N}}-i}10^{i-1}) & \text{if } b=1 \text{ and } i\neq h(X),\\
f(X,Z)=X& \text{otherwise.} 
\end{array}\right.
$$

The Markov chain is thus defined as 
$$
X_t= f(X_{t-1},Z_t)
$$


An integer $\ell\in \llbracket 1,{\mathsf{N}} \rrbracket$ is said {\it fair} 
at time $t$ if there
exists $0\leq j <t$ such that $Z_{j+1}=(\ell,\cdot)$ and $h(X_j)\neq \ell$.
In other words, there exists a date $j$ before $t$ where 
the first element of the random variable $Z$ is exactly $l$ 
(\textit{i.e.}, $l$ is the strategy at date $j$) 
and where the configuration $X_j$ allows to cross the edge $l$.  
 
Let $\ts$ be the first time all the elements of $\llbracket 1, {\mathsf{N}} \rrbracket$
are fair. The integer $\ts$ is a randomized stopping time for
the Markov chain $(X_t)$.

\begin{lmm}
The integer $\ts$ is a strong stationary time.
\end{lmm}

\begin{proof}
Let $\tau_\ell$ be the first time that $\ell$ is fair. The random variable
$Z_{\tau_\ell}$ is of the form $(\ell,b)$ 
such that 
$b=1$ with probability $\frac{1}{2}$ and $b=0$ with probability
$\frac{1}{2}$. Since $h(X_{\tau_\ell-1})\neq\ell$ the value of the $\ell$-th
bit of $X_{\tau_\ell}$ 
is $0$ or $1$ with the same probability ($\frac{1}{2}$).
This probability is independent of the value of the other bits.

Moving next in the chain, at each step,
the $l$-th bit  is switched from $0$ to $1$ or from $1$ to $0$ each time with
the same probability. Therefore,  for $t\geq \tau_\ell$, the
$\ell$-th bit of $X_t$ is $0$ or $1$ with the same probability,  and
independently of the value of the other bits, proving the
lemma.\end{proof}

\begin{thrm} \label{prop:stop}
If $\ov{h}$ is bijective and square-free, then
$E[\ts]\leq 8{\mathsf{N}}^2+ 4{\mathsf{N}}\ln ({\mathsf{N}}+1)$. 
\end{thrm}

For each $X\in \Bool^{\mathsf{N}}$ and $\ell\in\llbracket 1,{\mathsf{N}}\rrbracket$, 
let $S_{X,\ell}$ be the
random variable that counts the number of steps 
from $X$ until we reach a configuration where
$\ell$ is fair. More formally
\[
\begin{array}{rcl}
S_{X,\ell}&=&\min \{t \geq 1\mid h(X_{t-1})\neq \ell\text{ and }Z_t=(\ell,.) \\
&& \qquad \text{ and } X_0=X\}.
\end{array}
\]


\begin{lmm}\label{prop:lambda}
Let $\ov{h}$ is a square-free bijective function. Then
for all $X$ and 
all $\ell$, 
the inequality 
$E[S_{X,\ell}]\leq 8{\mathsf{N}}^2$ is established.
\end{lmm}

\begin{proof}
For every $X$, every $\ell$, one has $\P(S_{X,\ell}\leq 2)\geq
\frac{1}{4{\mathsf{N}}^2}$. 
Let $X_0= X$.
Indeed, 
\begin{itemize}
\item if $h(X)\neq \ell$, then
$\P(S_{X,\ell}=1)=\frac{1}{2{\mathsf{N}}}\geq \frac{1}{4{\mathsf{N}}^2}$. 
\item otherwise, $h(X)=\ell$, then
$\P(S_{X,\ell}=1)=0$.
But in this case, intuitively, it is possible to move
from $X$ to $\ov{h}^{-1}(X)$ (with probability $\frac{1}{2N}$). And in
$\ov{h}^{-1}(X)$ the $l$-th bit can be switched. 
More formally,
since $\ov{h}$ is square-free,
$\ov{h}(X)=\ov{h}(\ov{h}(\ov{h}^{-1}(X)))\neq \ov{h}^{-1}(X)$. It follows
that $(X,\ov{h}^{-1}(X))\in E_h$. We thus have
$P(X_1=\ov{h}^{-1}(X))=\frac{1}{2{\mathsf{N}}}$. Now, by Lemma~\ref{lm:h},
$h(\ov{h}^{-1}(X))\neq h(X)$. Therefore $\P(S_{x,\ell}=2\mid
X_1=\ov{h}^{-1}(X))=\frac{1}{2{\mathsf{N}}}$, proving that $\P(S_{x,\ell}\leq 2)\geq
\frac{1}{4{\mathsf{N}}^2}$.
\end{itemize}

Therefore, $\P(S_{X,\ell}\geq 3)\leq 1-\frac{1}{4{\mathsf{N}}^2}$. By induction, one
has, for every $i$, $\P(S_{X,\ell}\geq 2i)\leq
\left(1-\frac{1}{4{\mathsf{N}}^2}\right)^i$.
 Moreover,
since $S_{X,\ell}$ is positive, it is known~\cite[lemma 2.9]{proba}, that
$$E[S_{X,\ell}]=\sum_{i=1}^{+\infty}\P(S_{X,\ell}\geq i).$$
Since $\P(S_{X,\ell}\geq i)\geq \P(S_{X,\ell}\geq i+1)$, one has
\[
\begin{array}{rcl}
  E[S_{X,\ell}]&=&\sum_{i=1}^{+\infty}\P(S_{X,\ell}\geq i)\\
&\leq& 
\P(S_{X,\ell}\geq 1) +\P(S_{X,\ell}\geq 2)\\
&& \qquad +2 \sum_{i=1}^{+\infty}\P(S_{X,\ell}\geq 2i).
\end{array}
\]
Consequently,
$$E[S_{X,\ell}]\leq 1+1+2
\sum_{i=1}^{+\infty}\left(1-\frac{1}{4{\mathsf{N}}^2}\right)^i=2+2(4{\mathsf{N}}^2-1)=8{\mathsf{N}}^2,$$
which concludes the proof.
\end{proof}

Let $\ts^\prime$ be the time used to get all the bits but one fair.

\begin{lmm}\label{lm:stopprime}
One has $E[\ts^\prime]\leq 4{\mathsf{N}} \ln ({\mathsf{N}}+1).$
\end{lmm}

\begin{proof}
This is a classical  Coupon Collector's like problem. Let $W_i$ be the
random variable counting the number of moves done in the Markov chain while
we had exactly $i-1$ fair bits. One has $\ts^\prime=\sum_{i=1}^{{\mathsf{N}}-1}W_i$.
 But when we are at position $X$ with $i-1$ fair bits, the probability of
 obtaining a new fair bit is either $1-\frac{i-1}{{\mathsf{N}}}$ if $h(X)$ is fair,
 or  $1-\frac{i-2}{{\mathsf{N}}}$ if $h(X)$ is not fair. 

Therefore,
$\P (W_i=k)\leq \left(\frac{i-1}{{\mathsf{N}}}\right)^{k-1} \frac{{\mathsf{N}}-i+2}{{\mathsf{N}}}.$
Consequently, we have $\P(W_i\geq k)\leq \left(\frac{i-1}{{\mathsf{N}}}\right)^{k-1} \frac{{\mathsf{N}}-i+2}{{\mathsf{N}}-i+1}.$
It follows that $E[W_i]=\sum_{k=1}^{+\infty} \P (W_i\geq k)\leq {\mathsf{N}} \frac{{\mathsf{N}}-i+2}{({\mathsf{N}}-i+1)^2}\leq \frac{4{\mathsf{N}}}{{\mathsf{N}}-i+2}$.

It follows that 
$E[W_i]\leq \frac{4{\mathsf{N}}}{{\mathsf{N}}-i+2}$. Therefore
$$E[\ts^\prime]=\sum_{i=1}^{{\mathsf{N}}-1}E[W_i]\leq 
4{\mathsf{N}}\sum_{i=1}^{{\mathsf{N}}-1} \frac{1}{{\mathsf{N}}-i+2}=
4{\mathsf{N}}\sum_{i=3}^{{\mathsf{N}}+1}\frac{1}{i}.$$

But $\sum_{i=1}^{{\mathsf{N}}+1}\frac{1}{i}\leq 1+\ln({\mathsf{N}}+1)$. It follows that
$1+\frac{1}{2}+\sum_{i=3}^{{\mathsf{N}}+1}\frac{1}{i}\leq 1+\ln({\mathsf{N}}+1).$
Consequently,
$E[\ts^\prime]\leq 
4{\mathsf{N}} (-\frac{1}{2}+\ln({\mathsf{N}}+1))\leq 
4{\mathsf{N}}\ln({\mathsf{N}}+1)$.
\end{proof}

One can now prove Theorem~\ref{prop:stop}.

\begin{proof}
Since $\ts^\prime$ is the time used to obtain $\mathsf{N}-1$ fair bits.
Assume that the last unfair bit is $\ell$. One has
$\ts=\ts^\prime+S_{X_\tau,\ell}$, and therefore $E[\ts] =
E[\ts^\prime]+E[S_{X_\tau,\ell}]$. Therefore, Theorem~\ref{prop:stop} is a
direct application of Lemma~\ref{prop:lambda} and~\ref{lm:stopprime}.
\end{proof}

Now using Markov Inequality, one has $\P_X(\tau > t)\leq \frac{E[\tau]}{t}$.
With $t_n=32N^2+16N\ln (N+1)$, one obtains:  $\P_X(\tau > t_n)\leq \frac{1}{4}$. 
Therefore, using the definition of $t_{\rm mix}$ and
Theorem~\ref{thm-sst}, it follows that
$t_{\rm mix}\leq 32N^2+16N\ln (N+1)=O(N^2)$.

Notice that the calculus of the stationary time upper bound is obtained
under the following constraint: for each vertex in the $\mathsf{N}$-cube 
there are one ongoing arc and one outgoing arc that are removed. 
The calculus doesn't consider (balanced) Hamiltonian cycles, which 
are more regular and more binding than this constraint.
Moreover, the bound
is obtained using the coarse Markov Inequality. For the
classical (lazy) random walk the  $\mathsf{N}$-cube, without removing any
Hamiltonian cycle, the mixing time is in $\Theta(N\ln N)$. 
We conjecture that in our context, the mixing time is also in $\Theta(N\ln
N)$.

In this latter context, we claim that the upper bound for the stopping time 
should be reduced. This fact is studied in the next section.

\subsection{Practical Evaluation of Stopping Times}\label{sub:stop:exp}
 
Let be given a function $f: \Bool^{\mathsf{N}} \rightarrow \Bool^{\mathsf{N}}$
and an initial seed $x^0$.
The pseudo code given in Algorithm~\ref{algo:stop} returns the smallest 
number of iterations such that all elements $\ell\in \llbracket 1,{\mathsf{N}} \rrbracket$ are fair. It allows to deduce an approximation of $E[\ts]$
by calling this code many times with many instances of function and many 
seeds.

\begin{algorithm}[ht]
\KwIn{a function $f$, an initial configuration $x^0$ ($\mathsf{N}$ bits)}
\KwOut{a number of iterations $\textit{nbit}$}

$\textit{nbit} \leftarrow 0$\;
$x\leftarrow x^0$\;
$\textit{fair}\leftarrow\emptyset$\;
\While{$\left\vert{\textit{fair}}\right\vert < \mathsf{N} $}
{
        $ s \leftarrow \textit{Random}(\mathsf{N})$ \;
        $\textit{image} \leftarrow f(x) $\;
        \If{$\textit{Random}(1) \neq 0$ and $x[s] \neq \textit{image}[s]$}{
            $\textit{fair} \leftarrow \textit{fair} \cup \{s\}$\;
            $x[s] \leftarrow \textit{image}[s]$\;
          }
        $\textit{nbit} \leftarrow \textit{nbit}+1$\;
}
\Return{$\textit{nbit}$}\;
\caption{Pseudo Code of stopping time computation}
\label{algo:stop}
\end{algorithm}

Practically speaking, for each number $\mathsf{N}$, $ 3 \le \mathsf{N} \le 16$, 
10 functions have been generated according to the method presented in Section~\ref{sec:hamilton}. For each of them, the calculus of the approximation of $E[\ts]$
is executed 10000 times with a random seed. Figure~\ref{fig:stopping:moy}
summarizes these results. A circle represents the 
approximation of $E[\ts]$ for a given $\mathsf{N}$.
The line is the graph of the function $x \mapsto 2x\ln(2x+8)$. 
It can firstly 
be observed that the approximation is largely
smaller than the upper bound given in Theorem~\ref{prop:stop}.
It can be further deduced  that the conjecture of the previous section 
is realistic according to the graph of $x \mapsto 2x\ln(2x+8)$.


\begin{figure}
\centering
\includegraphics[width=0.49\textwidth]{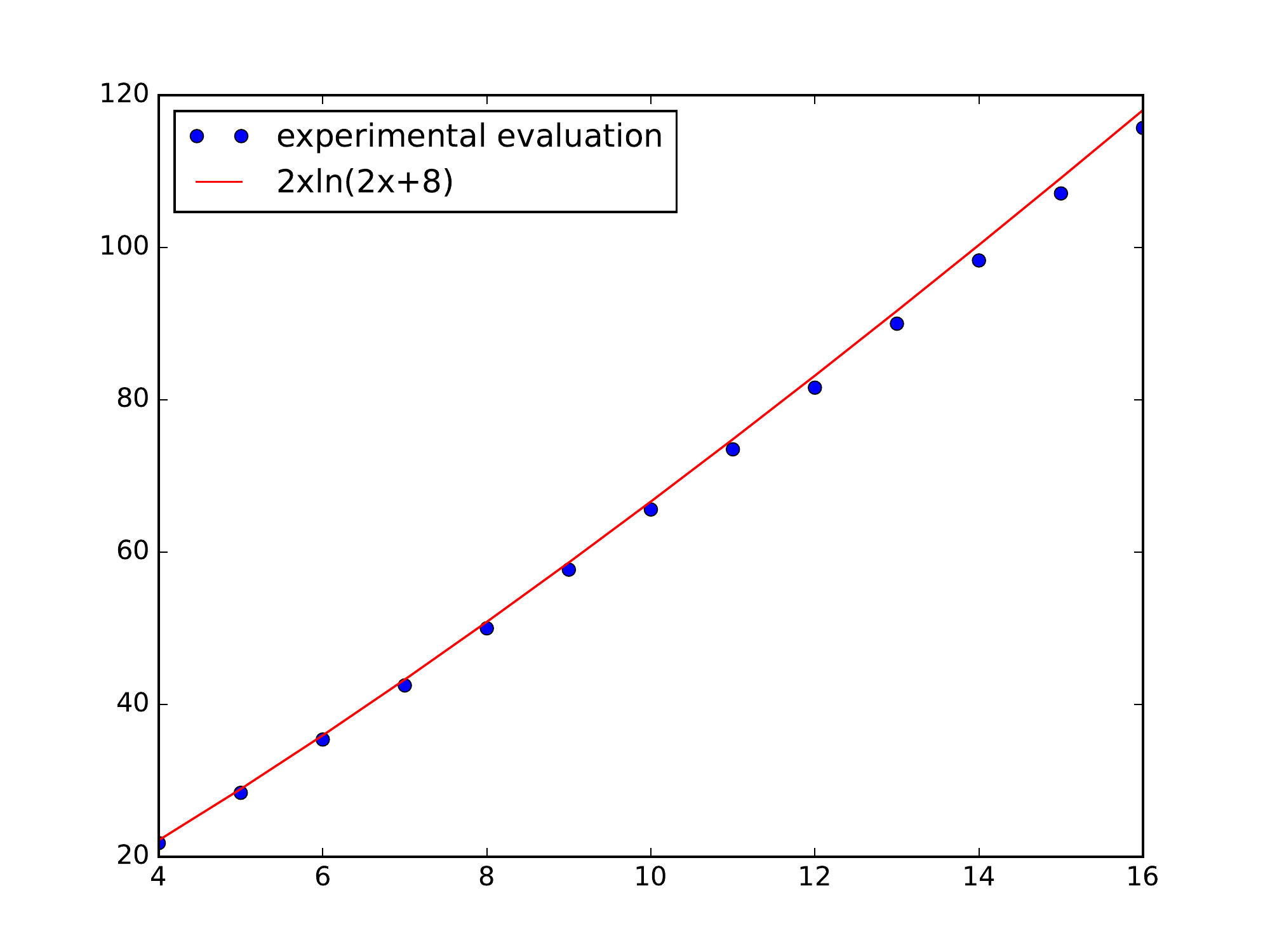}
\caption{Average Stopping Time Approximation}\label{fig:stopping:moy}
\end{figure}

\section{Experiments}\label{sec:prng}
Let us finally present the pseudorandom number generator $\chi_{\textit{16HamG}}$,
which is based on random walks in $\Gamma_{\{b\}}(f)$. 
More precisely, let be given a Boolean map $f:\Bool^{\mathsf{N}} \rightarrow 
\Bool^{\mathsf{N}}$,
a PRNG \textit{Random},
an integer $b$ that corresponds to an iteration number (\textit{i.e.}, the length of the walk), and 
an initial configuration $x^0$. 
Starting from $x^0$, the algorithm repeats $b$ times 
a random choice of which edge to follow, and crosses this edge 
provided it is allowed to do so, \textit{i.e.}, 
when $\textit{Random}(1)$ is not null. 
The final configuration is thus outputted.
This PRNG is formalized in Algorithm~\ref{CI Algorithm:2}.

\begin{algorithm}[ht]
\KwIn{a function $f$, an iteration number $b$, an initial configuration $x^0$ ($\mathsf{N}$ bits)}
\KwOut{a configuration $x$ ($\mathsf{N}$ bits)}
$x\leftarrow x^0$\;
\For{$i=0,\dots,b-1$}
{
\If{$\textit{Random}(1) \neq 0$}{
$s^0\leftarrow{\textit{Random}(\mathsf{N})}$\;
$x\leftarrow{F_f(x,s^0)}$\;
}
}
return $x$\;
\caption{Pseudo Code of the $\chi_{\textit{16HamG}}$ PRNG}
\label{CI Algorithm:2}
\end{algorithm}

This PRNG is slightly different from $\chi_{\textit{14Secrypt}}$
recalled in Algorithm~\ref{CI Algorithm}.
As this latter, the length of the random 
walk of our algorithm is always constant (and is equal to $b$). 
However, in the current version, we add the constraint that   
the probability to execute the function $F_f$ is equal to 0.5 since
the output of $\textit{Random(1)}$ is uniform in $\{0,1\}$.  
This constraint is added to match the theoretical framework of 
Sect.~\ref{sec:hypercube}.

Notice that the chaos property of $G_f$ given in Sect.\ref{sec:proofOfChaos}
only requires the graph $\Gamma_{\{b\}}(f)$ to be  strongly connected.
Since the $\chi_{\textit{16HamG}}$ algorithm 
only adds probability constraints on existing edges, 
it preserves this property.

For each number $\mathsf{N}=4,5,6,7,8$ of bits, we have generated 
the functions according to the method 
given in Sect.~\ref{sec:SCCfunc} and~\ref{sec:hamilton}. 
For each $\mathsf{N}$, we have then restricted this evaluation to the function 
whose Markov Matrix (issued from Eq.~(\ref{eq:Markov:rairo})) 
has the smallest practical mixing time.
Such functions are 
given in Table~\ref{table:nc}.
In this table, let us consider, for instance, 
the function $\textcircled{a}$ from $\Bool^4$ to $\Bool^4$
defined by the following images : 
$[13, 10, 9, 14, 3, 11, 1, 12, 15, 4, 7, 5, 2, 6, 0, 8]$.
In other words,  the image of $3~(0011)$ by $\textcircled{a}$ is $14~(1110)$:
it is obtained as  the  binary  value  of  the  fourth element  in  
the  second  list (namely~14).  

In this table the column that is labeled with $b$ 
gives the practical mixing time 
where the deviation to the standard distribution is inferior than $10^{-6}$.

\begin{table*}[t]
\begin{center}
\begin{scriptsize}
\begin{tabular}{|c|c|c|c|}
\hline
Function $f$ & $f(x)$, for $x$ in $(0,1,2,\hdots,2^n-1)$ & $\mathsf{N}$ & $b$ 
\\ 
\hline
$\textcircled{a}$&[13,10,9,14,3,11,1,12,15,4,7,5,2,6,0,8]&4&64\\
\hline
$\textcircled{b}$& 
[29, 22, 25, 30, 19, 27, 24, 16, 21, 6, 5, 28, 23, 26, 1, 17, & 5 & 78 \\
&
 31, 12, 15, 8, 10, 14, 13, 9, 3, 2, 7, 20, 11, 18, 0, 4]
&&\\
\hline
&
[55, 60, 45, 44, 58, 62, 61, 48, 53, 50, 52, 36, 59, 34, 33, 49,
&&\\
&
 15, 42, 47, 46, 35, 10, 57, 56, 7, 54, 39, 37, 51, 2, 1, 40, 63,
&&\\
$\textcircled{c}$&
 26, 25, 30, 19, 27, 17, 28, 31, 20, 23, 21, 18, 22, 16, 24, 13, 
&6&88\\
&
12, 29, 8, 43, 14, 41, 0, 5, 38, 4, 6, 11, 3, 9, 32]
&&\\
\hline
&
[111, 124, 93, 120, 122, 114, 89, 121, 87, 126, 125, 84, 123, 82, 
&&\\
&112, 80, 79, 106, 105, 110, 75, 107, 73, 108, 119, 100, 117, 116, 
&&\\
&103, 102, 101, 97, 31, 86, 95, 94, 83, 26, 88, 24, 71, 118, 69, 
&&\\
&68, 115, 90, 113, 16, 15, 76, 109, 72, 74, 10, 9, 104, 7, 6, 65, 
&&\\
$\textcircled{d}$ &70, 99, 98, 64, 96, 127, 54, 53, 62, 51, 59, 56, 60, 39, 52, 37, &7 &99\\
&36, 55, 58, 57, 49, 63, 44, 47, 40, 42, 46, 45, 41, 35, 34, 33, 
&&\\
&38, 43, 50, 32, 48, 29, 28, 61, 92, 91, 18, 17, 25, 19, 30, 85, 
&&\\
&22, 27, 2, 81, 0, 13, 78, 77, 14, 3, 11, 8, 12, 23, 4, 21, 20, 
&&\\
&67, 66, 5, 1]
&&\\

\hline
&
[223, 238, 249, 254, 243, 251, 233, 252, 183, 244, 229, 245, 227, 
&&\\
&246, 240, 176, 175, 174, 253, 204, 203, 170, 169, 248, 247, 226, 
&&\\
&228, 164, 163, 162, 161, 192, 215, 220, 205, 216, 155, 222, 221, 
&&\\
&208, 213, 150, 212, 214, 219, 211, 145, 209, 239, 202, 207, 140, 
&&\\
&195, 234, 193, 136, 231, 230, 199, 197, 131, 198, 225, 200, 63, 
&&\\
&188, 173, 184, 186, 250, 57, 168, 191, 178, 180, 52, 187, 242, 
&&\\
&241, 48, 143, 46, 237, 236, 235, 138, 185, 232, 135, 38, 181, 165, 
&&\\
&35, 166, 33, 224, 31, 30, 153, 158, 147, 218, 217, 156, 159, 148, 
&&\\
$\textcircled{e}$&151, 149, 19, 210, 144, 152, 141, 206, 13, 12, 171, 10, 201, 128, 
&8&109\\
&133, 130, 132, 196, 3, 194, 137, 0, 255, 124, 109, 120, 122, 106, 
&&\\
&125, 104, 103, 114, 116, 118, 123, 98, 97, 113, 79, 126, 111, 110, 
&&\\
&99, 74, 121, 72, 71, 70, 117, 101, 115, 102, 65, 112, 127, 90, 89, 
&&\\
&94, 83, 91, 81, 92, 95, 84, 87, 85, 82, 86, 80, 88, 77, 76, 93, 
&&\\
&108, 107, 78, 105, 64, 69, 66, 68, 100, 75, 67, 73, 96, 55, 190, 
&&\\
&189, 62, 51, 59, 41, 60, 119, 182, 37, 53, 179, 54, 177, 32, 45, 
&&\\
&44, 61, 172, 11, 58, 9, 56, 167, 34, 36, 4, 43, 50, 49, 160, 23, 
&&\\
&28, 157, 24, 26, 154, 29, 16, 21, 18, 20, 22, 27, 146, 25, 17, 47, 
&&\\
&142, 15, 14, 139, 42, 1, 40, 39, 134, 7, 5, 2, 6, 129, 8]
&&\\
\hline
\end{tabular}
\end{scriptsize}
\end{center}
\caption{Functions with DSCC Matrix and smallest MT}\label{table:nc}
\end{table*}

Let us first discuss about results against the NIST test suite. 
In our experiments, 100 sequences (s = 100) of 1,000,000 bits are generated and tested.
If the value $\mathbb{P}_T$ of any test is smaller than 0.0001, the sequences are considered to be not good enough
and the generator is unsuitable. 

Table~\ref{The passing rate} shows $\mathbb{P}_T$ of sequences based 
on $\chi_{\textit{16HamG}}$ using different functions, namely
$\textcircled{a}$,\ldots, $\textcircled{e}$.
In this algorithm implementation, 
the embedded PRNG \textit{Random} is the default Python PRNG, \textit{i.e.},
the Mersenne Twister algorithm~\cite{matsumoto1998mersenne}. 
Implementations for $\mathsf{N}=4, \dots, 8$ of this algorithm is evaluated
through the NIST test suite and results are given in columns 
$\textit{MT}_4$, \ldots,  $\textit{MT}_8$.
If there are at least two statistical values in a test, this test is
marked with an asterisk and the average value is computed to characterize the statistics.

We first can see in Table \ref{The passing rate} that all the rates 
are greater than 97/100, \textit{i.e.}, all the generators 
achieve to pass the NIST battery of tests.
It can be noticed that adding chaos properties for Mersenne Twister 
algorithm does not reduce its security against this statistical tests.

\begin{table*} 
\renewcommand{\arraystretch}{1.1}
\begin{center}
\begin{tiny}
\setlength{\tabcolsep}{2pt}

\begin{tabular}{|l|r|r|r|r|r|}
 \hline 
Test & $\textit{MT}_4$ & $\textit{MT}_5$& $\textit{MT}_6$& $\textit{MT}_7$& $\textit{MT}_8$
 \\ \hline 
Frequency (Monobit)& 0.924 (1.0)& 0.678 (0.98)& 0.102 (0.97)& 0.213 (0.98)& 0.719 (0.99) \\ \hline 
Frequency  within a Block& 0.514 (1.0)& 0.419 (0.98)& 0.129 (0.98)& 0.275 (0.99)& 0.455 (0.99)\\ \hline 
Cumulative Sums (Cusum) *& 0.668 (1.0)& 0.568 (0.99)& 0.881 (0.98)& 0.529 (0.98)& 0.657 (0.995)\\ \hline 
Runs& 0.494 (0.99)& 0.595 (0.97)& 0.071 (0.97)& 0.017 (1.0)& 0.834 (1.0)\\ \hline 
Longest Run of Ones in a Block& 0.366 (0.99)& 0.554 (1.0)& 0.042 (0.99)& 0.051 (0.99)& 0.897 (0.97)\\ \hline 
Binary Matrix Rank& 0.275 (0.98)& 0.494 (0.99)& 0.719 (1.0)& 0.334 (0.98)& 0.637 (0.99)\\ \hline 
Discrete Fourier Transform (Spectral)& 0.122 (0.98)& 0.108 (0.99)& 0.108 (1.0)& 0.514 (0.99)& 0.534 (0.98)\\ \hline 
Non-overlapping Template Matching*& 0.483 (0.990)& 0.507 (0.990)& 0.520 (0.988)& 0.494 (0.988)& 0.515 (0.989)\\ \hline 
Overlapping Template Matching& 0.595 (0.99)& 0.759 (1.0)& 0.637 (1.0)& 0.554 (0.99)& 0.236 (1.0)\\ \hline 
Maurer's "Universal Statistical"& 0.202 (0.99)& 0.000 (0.99)& 0.514 (0.98)& 0.883 (0.97)& 0.366 (0.99)\\ \hline 
Approximate Entropy (m=10)& 0.616 (0.99)& 0.145 (0.99)& 0.455 (0.99)& 0.262 (0.97)& 0.494 (1.0)\\ \hline 
Random Excursions *& 0.275 (1.0)& 0.495 (0.975)& 0.465 (0.979)& 0.452 (0.991)& 0.260 (0.989)\\ \hline 
Random Excursions Variant *& 0.382 (0.995)& 0.400 (0.994)& 0.417 (0.984)& 0.456 (0.991)& 0.389 (0.991)\\ \hline 
Serial* (m=10)& 0.629 (0.99)& 0.963 (0.99)& 0.366 (0.995)& 0.537 (0.985)& 0.253 (0.995)\\ \hline 
Linear Complexity& 0.494 (0.99)& 0.514 (0.98)& 0.145 (1.0)& 0.657 (0.98)& 0.145 (0.99)\\ \hline 
\end{tabular}

\begin{tabular}{|l|r|r|r|r|r|}
 \hline 
Test  
&$\textcircled{a}$& $\textcircled{b}$ & $\textcircled{c}$ & $\textcircled{d}$ & $\textcircled{e}$ \\ \hline 
Frequency (Monobit)&0.129 (1.0)& 0.181 (1.0)& 0.637 (0.99)& 0.935 (1.0)& 0.978 (1.0)\\ \hline 
Frequency  within a Block& 0.275 (1.0)& 0.534 (0.98)& 0.066 (1.0)& 0.719 (1.0)& 0.366 (1.0)\\ \hline 
Cumulative Sums (Cusum) *& 0.695 (1.0)& 0.540 (1.0)& 0.514 (0.985)& 0.773 (0.995)& 0.506 (0.99)\\ \hline 
Runs&  0.897 (0.99)& 0.051 (1.0)& 0.102 (0.98)& 0.616 (0.99)& 0.191 (1.0)\\ \hline 
Longest Run of Ones in a Block&  0.851 (1.0)& 0.595 (0.99)& 0.419 (0.98)& 0.616 (0.98)& 0.897 (1.0)\\ \hline 
Binary Matrix Rank& 0.419 (1.0)& 0.946 (0.99)& 0.319 (0.99)& 0.739 (0.97)& 0.366 (1.0)\\ \hline 
Discrete Fourier Transform (Spectral)&  0.867 (1.0)& 0.514 (1.0)& 0.145 (1.0)& 0.224 (0.99)& 0.304 (1.0)\\ \hline 
Non-overlapping Template Matching*& 0.542 (0.990)& 0.512 (0.989)& 0.505 (0.990)& 0.494 (0.989)& 0.493 (0.991)\\ \hline 
Overlapping Template Matching&  0.275 (0.99)& 0.080 (0.99)& 0.574 (0.98)& 0.798 (0.99)& 0.834 (0.99)\\ \hline 
Maurer's "Universal Statistical"&  0.383 (0.99)& 0.991 (0.98)& 0.851 (1.0)& 0.595 (0.98)& 0.514 (1.0)\\ \hline 
Approximate Entropy (m=10)&  0.935 (1.0)& 0.719 (1.0)& 0.883 (1.0)& 0.719 (0.97)& 0.366 (0.99)\\ \hline 
Random Excursions *& 0.396 (0.991)& 0.217 (0.989)& 0.445 (0.975)& 0.743 (0.993)& 0.380 (0.990)\\ \hline 
Random Excursions Variant *& 0.486 (0.997)& 0.373 (0.981)& 0.415 (0.994)& 0.424 (0.991)& 0.380 (0.991)\\ \hline 
Serial* (m=10)&0.350 (1.0)& 0.678 (0.995)& 0.287 (0.995)& 0.740 (0.99)& 0.301 (0.98)\\ \hline 
Linear Complexity& 0.455 (0.99)& 0.867 (1.0)& 0.401 (0.99)& 0.191 (0.97)& 0.699 (1.0)\\ \hline 
\end{tabular}

\end{tiny}
\end{center}
\caption{NIST SP 800-22 test results ($\mathbb{P}_T$)}
\label{The passing rate}
\end{table*}


\section{Conclusion}
This work has assumed a Boolean map $f$ which is embedded into   
a discrete-time dynamical system $G_f$.
This one is supposed to be iterated a fixed number 
$p_1$ or $p_2$,\ldots, or $p_{\mathds{p}}$ 
times before its output is considered. 
This work has first shown that iterations of
$G_f$ are chaotic if and only if its iteration graph $\Gamma_{\mathcal{P}}(f)$
is strongly connected where $\mathcal{P}$ is $\{p_1, \ldots, p_{\mathds{p}}\}$.
It can be deduced that in such a situation a PRNG, which iterates $G_f$,
satisfies the property of chaos and can be used in simulating chaos 
phenomena.

We then have shown that a previously presented approach can be directly 
applied here to generate function $f$ with strongly connected 
$\Gamma_{\mathcal{P}}(f)$. 
The iterated map inside the generator is built by first removing from a 
$\mathsf{N}$-cube a balanced  Hamiltonian cycle and next 
by adding  a self loop to each vertex. 
The PRNG can thus be seen as a random walk of length in $\mathcal{P}$
into  this new $\mathsf{N}$-cube.
We have presented an efficient method to compute such a balanced Hamiltonian 
cycle. This method is an algebraic solution of an undeterministic 
approach~\cite{ZanSup04} and has a low complexity.
To the best of the authors knowledge, this is the first time a full 
automatic method to provide chaotic PRNGs is given.  
Practically speaking, this approach preserves the security properties of 
the embedded PRNG, even if it remains quite cost expensive.

We furthermore have presented an upper bound on the number of iterations 
that is sufficient to obtain an uniform distribution of the output.
Such an upper bound is quadratic on the number of bits to output.
Experiments have however shown that such a bound is in 
$\mathsf{N}.\log(\mathsf{N})$ in practice.
Finally,  experiments through the  NIST battery have shown that
the statistical properties are almost established for
 $\mathsf{N} = 4, 5, 6, 7, 8$ and should be observed for any 
positive integer $\mathsf{N}$.

In future work, we intend to understand the link between 
statistical tests and the properties of chaos for
the associated iterations.
By doing so, relations between desired statistically unbiased behaviors and
topological properties will be understood, leading to better choices
in iteration functions. 
Conditions allowing the reduction of the stopping-time will be
investigated too, while other modifications of the hypercube will
be regarded in order to enlarge the set of known chaotic
and random iterations.


\section*{Acknowledgements}
This work is partially funded by the Labex ACTION program (contract ANR-11-LABX-01-01). 
Computations presented in this article were realised on the supercomputing
facilities provided by the M\'esocentre de calcul de Franche-Comt\'e.

\bibliographystyle{plain} 
\bibliography{biblio}

\begin{thebibliography}{10}

\bibitem{bcgr11:ip}
Jacques Bahi, Jean-Fran\c{c}ois Couchot, Christophe Guyeux, and Adrien Richard.
\newblock On the link between strongly connected iteration graphs and chaotic
  boolean discrete-time dynamical systems.
\newblock In {\em FCT'11, 18th Int. Symp. on Fundamentals of Computation
  Theory}, volume 6914 of {\em LNCS}, pages 126--137, Oslo, Norway, August
  2011.

\bibitem{bfgw11:ip}
Jacques Bahi, Xiaole Fang, Christophe Guyeux, and Qianxue Wang.
\newblock On the design of a family of {CI} pseudo-random number generators.
\newblock In {\em WICOM'11, 7th Int. IEEE Conf. on Wireless Communications,
  Networking and Mobile Computing}, pages 1--4, Wuhan, China, September 2011.

\bibitem{Banks92}
J.~Banks, J.~Brooks, G.~Cairns, and P.~Stacey.
\newblock On {D}evaney's definition of chaos.
\newblock {\em Amer. Math. Monthly}, 99:332--334, 1992.

\bibitem{Nist10}
E.~Barker and A.~Roginsky.
\newblock Draft {N}{I}{S}{T} special publication 800-131 recommendation for the
  transitioning of cryptographic algorithms and key sizes, 2010.

\bibitem{DBLP:journals/combinatorics/BhatS96}
Girish~S. Bhat and Carla~D. Savage.
\newblock Balanced gray codes.
\newblock {\em Electr. J. Comb.}, 3(1), 1996.

\bibitem{Bykov2016}
I.~S. Bykov.
\newblock On locally balanced gray codes.
\newblock {\em Journal of Applied and Industrial Mathematics}, 10(1):78--85,
  2016.

\bibitem{5376454}
Li~Cao, Lequan Min, and Hongyan Zang.
\newblock A chaos-based pseudorandom number generator and performance analysis.
\newblock In {\em Computational Intelligence and Security, 2009. CIS '09.
  International Conference on}, volume~1, pages 494--498. IEEE, Dec 2009.

\bibitem{DBLP:conf/secrypt/CouchotHGWB14}
Jean{-}Fran{\c{c}}ois Couchot, Pierre{-}Cyrille H{\'{e}}am, Christophe Guyeux,
  Qianxue Wang, and Jacques~M. Bahi.
\newblock Pseudorandom number generators with balanced gray codes.
\newblock In Mohammad~S. Obaidat, Andreas Holzinger, and Pierangela Samarati,
  editors, {\em {SECRYPT} 2014 - Proceedings of the 11th International
  Conference on Security and Cryptography, Vienna, Austria, 28-30 August,
  2014}, pages 469--475. SciTePress, 2014.

\bibitem{Devaney}
Robert~L. Devaney.
\newblock {\em An Introduction to Chaotic Dynamical Systems}.
\newblock Addison-Wesley, Redwood City, CA, 2nd edition, 1989.

\bibitem{guyeuxTaiwan10}
Christophe Guyeux, Qianxue Wang, and J.M. Bahi.
\newblock Improving random number generators by chaotic iterations application
  in data hiding.
\newblock In {\em Computer Application and System Modeling (ICCASM), 2010
  International Conference on}, volume~13, pages V13--643--V13--647. IEEE, Oct
  2010.

\bibitem{LEcuyerS07}
Pierre L'Ecuyer and Richard~J. Simard.
\newblock Test{U01}: {A} {C} library for empirical testing of random number
  generators.
\newblock {\em ACM Trans. Math. Softw}, 33(4), 2007.

\bibitem{LevinPeresWilmer2006}
David~A. Levin, Yuval Peres, and Elizabeth~L. Wilmer.
\newblock {\em {Markov chains and mixing times}}.
\newblock American Mathematical Society, 2006.

\bibitem{Marsaglia1996}
G.~Marsaglia.
\newblock Diehard: a battery of tests of randomness.
\newblock {\em http://stat.fsu.edu/~geo/diehard.html}, 1996.

\bibitem{matsumoto1998mersenne}
Makoto Matsumoto and Takuji Nishimura.
\newblock Mersenne twister: a 623-dimensionally equidistributed uniform
  pseudo-random number generator.
\newblock {\em ACM Transactions on Modeling and Computer Simulation (TOMACS)},
  8(1):3--30, 1998.

\bibitem{proba}
M.~Mitzenmacher and Eli Upfal.
\newblock {\em Probability and Computing}.
\newblock Cambridge University Press, 2005.

\bibitem{Robinson:1981:CS}
John~P. Robinson and Martin Cohn.
\newblock Counting sequences.
\newblock {\em IEEE Trans. Comput.}, 30(1):17--23, January 1981.

\bibitem{915385}
T.~Stojanovski and L.~Kocarev.
\newblock Chaos-based random number generators-part i: analysis [cryptography].
\newblock {\em Circuits and Systems I: Fundamental Theory and Applications,
  IEEE Transactions on}, 48(3):281--288, Mar 2001.

\bibitem{915396}
T.~Stojanovski, J.~Pihl, and L.~Kocarev.
\newblock Chaos-based random number generators. part ii: practical realization.
\newblock {\em Circuits and Systems I: Fundamental Theory and Applications,
  IEEE Transactions on}, 48(3):382--385, Mar 2001.

\bibitem{ZanSup04}
IN~Suparta and AJ~van Zanten.
\newblock Totally balanced and exponentially balanced gray codes.
\newblock {\em Discrete Analysis and Operation Research (Russia)},
  11(4):81--98, 2004.

\bibitem{wbg10ip}
Qianxue Wang, Jacques Bahi, Christophe Guyeux, and Xiaole Fang.
\newblock Randomness quality of {CI} chaotic generators. application to
  internet security.
\newblock In {\em INTERNET'2010. The 2nd Int. Conf. on Evolving Internet},
  pages 125--130, Valencia, Spain, September 2010. IEEE Computer Society Press.
\newblock Best Paper award.

\bibitem{wbg10:ip}
Qianxue Wang, Jacques Bahi, Christophe Guyeux, and Xiaole Fang.
\newblock Randomness quality of {CI} chaotic generators. application to
  internet security.
\newblock In {\em INTERNET'2010. The 2nd Int. Conf. on Evolving Internet},
  pages 125--130, Valencia, Spain, September 2010. IEEE Computer Society Press.
\newblock Best Paper award.

\end{thebibliography}

\end{document}